\documentclass[11pt]{article}

\usepackage[margin=1in]{geometry}
\usepackage[T1]{fontenc}
\usepackage[utf8]{inputenc}
\usepackage{lmodern}
\usepackage{microtype}
\usepackage{amsmath,amssymb,amsthm,mathtools}
\usepackage{thm-restate}
\makeatletter
\IfFormatAtLeastTF{2026-06-01}{%
  \renewcommand\thmt@autorefsetup{%
    \@xa\def\csname\thmt@envname autorefname\@xa\endcsname\@xa{\thmt@thmname}%
  }%
}{}
\makeatother
\usepackage{xcolor,booktabs}
\usepackage{xurl}
\usepackage{hyperref}
\usepackage[capitalize,noabbrev]{cleveref}

\hypersetup{
  colorlinks=true,
  linkcolor=blue!55!black,
  citecolor=green!45!black,
  urlcolor=blue!60!black,
  pdfauthor={Debarati Das, Evangelos Kipouridis, Tomasz Kociumaka},
  pdftitle={Metric Weighted Edit Distance: A (3+epsilon)-Approximation in O-tilde\string_epsilon(N\string^1.6) Time}
}

\newtheorem{theorem}{Theorem}[section]
\newtheorem{lemma}[theorem]{Lemma}
\newtheorem{proposition}[theorem]{Proposition}
\newtheorem{corollary}[theorem]{Corollary}
\newtheorem{observation}[theorem]{Observation}
\crefname{observation}{Observation}{Observations}
\Crefname{observation}{Observation}{Observations}
\theoremstyle{definition}

\theoremstyle{remark}

\newcommand{\gap}{\bot}
\newcommand{\dd}{\mathinner{.\,.}}
\newcommand{\WED}{\operatorname{WED}}
\newcommand{\abs}[1]{\lvert #1\rvert}
\newcommand{\Ohtilde}{\widetilde O}

\title{\bfseries\boldmath Metric Weighted Edit Distance:\\
A \((3+\varepsilon)\)-Approximation in \(\widetilde O_\varepsilon(N^{1.6})\) Time}
\author{%
  \begin{tabular}{@{}c@{\quad}c@{\quad}c@{}}
    Debarati Das & Evangelos Kipouridis & Tomasz Kociumaka\\
    \small Pennsylvania &
      \multicolumn{2}{c@{}}{\small Max Planck Institute for Informatics}\\
    \small State University & \multicolumn{2}{c@{}}{\small Saarland Informatics Campus}\\
    \footnotesize\href{mailto:dxd5606@psu.edu}{\nolinkurl{dxd5606@psu.edu}} &
    \footnotesize\href{mailto:kipouridis@mpi-inf.mpg.de}{\nolinkurl{kipouridis@mpi-inf.mpg.de}} &
    \footnotesize\href{mailto:tomasz.kociumaka@mpi-inf.mpg.de}{\nolinkurl{tomasz.kociumaka@mpi-inf.mpg.de}}
  \end{tabular}%
}
\date{}

\begin{document}

\maketitle
\vspace{-2em}

\begin{abstract}
For every \(0<\varepsilon\le1\), we give a randomized
\((3+\varepsilon)\)-approximation to weighted edit distance when
the costs form a metric on the alphabet augmented with a gap symbol.  For
strings of total length \(N\), the running time is
\(\widetilde O(N^{8/5}/\varepsilon^{16/5})\), where \(\widetilde O\) suppresses
factors polynomial in \(\log(N/\varepsilon)\).
The dependence on \(N\) matches that of the fastest known
\((3+\varepsilon)\)-approximation for unit-cost edit distance.
The algorithm never underestimates the edit distance and achieves
the approximation guarantee with inverse-polynomial failure probability
in \(N\).  The running time bound assumes constant-time exact arithmetic operations
and metric queries, and it is independent of the numerical range of the
edit costs.

We build on three tools: the sampling framework of Chakraborty, Das,
Goldenberg, Kouck\'y, and Saks (J.\ ACM, 2020), with subsequent
refinements by Andoni (2020); Kuszmaul's removal of inexpensive characters
(ICALP 2019); and Klein's data structure for distances in planar
graphs (SODA 2005).  Our new ingredients include, among others, a decomposition of
one string into pieces of bounded length with highly structured total
deletion costs.  This decomposition lets us compare all pieces against
a small family of substrings of the other string.
\end{abstract}

\section{Introduction}
\label{sec:introduction}

Edit distance is a basic tool in sequence comparison, with applications including
error-correcting codes and biological sequence analysis~\cite{Lev66,NeedlemanWunsch70}.
For strings \(X\) and \(Y\), their edit distance is the
minimum number of character insertions, deletions, and substitutions
that transform \(X\) into \(Y\).

Weighted costs reflect symbol similarities and the importance of
inserted or deleted symbols.  We study \emph{weighted edit distance}
for a \emph{metric} \(w\) on \(\Sigma\cup\{\gap\}\), where \(\Sigma\) is the alphabet and
\(\gap\notin\Sigma\) is a gap symbol.
Substituting \(b\) for \(a\) costs \(w(a,b)\), while inserting or
deleting \(a\) costs \(w(\gap,a)=w(a,\gap)\).
The weighted edit distance \(\WED_w(X,Y)\) is the minimum total
cost of an edit sequence transforming \(X\) into \(Y\).
Unit-cost edit distance assigns distance one to every pair of
distinct symbols, including the gap.
Write \(N\coloneqq |X|+|Y|\) for the total input length.

The classical \(O(N^2)\)-time dynamic-programming algorithm arose in several early
papers~\cite{NeedlemanWunsch70,Sankoff72,Sellers74,WagnerFischer74};
the latter two already allowed nonuniform edit
costs.
Even for unit costs, Backurs and Indyk~\cite{BI18} showed that a
\emph{truly subquadratic} exact algorithm, that is, an
\(O(N^{2-\delta})\)-time algorithm for any fixed
\(\delta>0\), would contradict the Strong Exponential Time Hypothesis
(SETH).  Bringmann and K\"unnemann~\cite{BringmannKunnemann15}
established the same barrier for binary strings.

\paragraph{Approximation algorithms.}
An \(\alpha\)-approximation returns an estimate between the distance
and \(\alpha\) times the distance.  Early algorithms for unit-cost
edit distance achieved approximation ratios polynomial in \(N\)
in near-linear time \(\Ohtilde(N)\)~\cite{LMS98,BJKK04,BES06}.\footnote{The notation \(\Ohtilde\)
suppresses factors polynomial in \(\log N\), or in
\(\log(N/\varepsilon)\) when an accuracy parameter
\(0<\varepsilon\le1\) is present.  Our $\Ohtilde(\cdot)$
bounds display every polynomial dependence on \(\varepsilon^{-1}\),
and their hidden constants are independent~of~\(\varepsilon\).}
Building on earlier work~\cite{OR07}, Andoni and Onak~\cite{AO12}
obtained a \(2^{O(\sqrt{\log N\log\log N})}\)-approximation
in \(N^{1+o(1)}\) time.
Andoni, Krauthgamer, and Onak~\cite{AKO10} then obtained a
\((\log N)^{O(1/\delta)}\)-approximation in \(N^{1+\delta}\) time for every
fixed \(\delta>0\).  Chakraborty, Das, Goldenberg, Kouck\'y, and
Saks~\cite{CDGKS20} achieved the first
classical constant-factor approximation in truly subquadratic time,
namely \(\Ohtilde(N^{12/7})\),
followed by simplified
variants of the same scheme~\cite{Andoni18,Andoni20}.
Goldenberg, Rubinstein, and Saha~\cite{GRS20} further improved the scheme, resulting in a
\((3+o(1))\)-approximation in \(N^{8/5+o(1)}\) time.
Constant-factor approximation in almost-linear \(N^{1+\delta}\) time was
first obtained for sufficiently large distances~\cite{BR20,KS20},
and then for all distances by Andoni and
Nosatzki~\cite{AndoniNosatzki20}, for every fixed \(\delta>0\),
with the approximation factor depending on \(\delta\).
Mao and Rubinstein~\cite{MR26} obtained a
\((1+\varepsilon)\)-approximation for every fixed accuracy
\(\varepsilon>0\) in \(N^2/2^{\log^{\Omega(1)}N}\) time.
The above approximation algorithms are randomized; see~\cite{KS23}
for a much weaker deterministic approximation algorithm.

For weighted edit distance, Kuszmaul~\cite{Kuszmaul19}
proved that, for every fixed \(0<\delta<1\), an
\(O(N^\delta)\)-approximation can be computed in
\(\Ohtilde(N^{2-\delta})\) time.  His result is stated for
equal-length strings and allows arbitrary metric costs.
No nontrivial approximations are known for the non-metric case.
\subsection*{Our result}

We give a randomized \((3+\varepsilon)\)-approximation for metric
weighted edit distance in
\(\Ohtilde(N^{8/5}/\varepsilon^{16/5})\) time.
The algorithm allows unequal lengths and arbitrary positive real
costs, with no dependence on their numerical range.  Its estimate is
always an upper bound; randomness affects only its accuracy.

For fixed \(\varepsilon\), the exponent \(8/5\) matches the best
known for \((3+\varepsilon)\)-approximation of ordinary edit
distance~\cite{GRS20}.  Since unit costs form the discrete metric, lowering
this exponent while preserving the approximation guarantee would
also improve the best known exponent for the unit-cost problem.

We use a real RAM with separate registers for exact reals and
\(O(\log(N/\varepsilon))\)-bit integers
\cite[Section~6.1]{EricksonHoogMiltzow24}.
Real addition, subtraction, multiplication, division, and
comparison take constant time; real registers do not support
floor, digit extraction, or conversion to integers.
We assume that each metric query and each independent unbiased random bit take
constant time.  Our bounds do not capture the bit cost of representing
and manipulating arbitrary reals.

For convenience, we assume \(N\ge2\); smaller inputs are handled
exactly in constant time.  An event holds \emph{with high probability} if its probability
is at least \(1-N^{-c}\) for any prescribed constant \(c>0\).
Increasing \(c\) changes only constant factors in the running time.

\begin{restatable}[Main result]{theorem}{mainresult}
\label{cor:weighted-cdgks-scale-selection}
Let \(X\) and \(Y\) be strings of total length \(N\), and let \(w\) be a
metric on their alphabet \(\Sigma\) together with the gap symbol \(\gap\).
There is a randomized real-RAM algorithm that, given $X,Y$, an accuracy parameter $0<\varepsilon\le 1$,
and oracle access to $w$, in \(\Ohtilde(N^{8/5}/\varepsilon^{16/5})\) time returns
a number \(Q\) with the following guarantees:
\begin{enumerate}
\item Deterministically, \(Q\ge\WED_w(X,Y)\).
\item With high probability, \(Q\le(3+\varepsilon)\WED_w(X,Y)\).
\end{enumerate}
\end{restatable}

For every \(0<\varepsilon\le1\), the same
\(\Ohtilde(N^{8/5}/\varepsilon^{16/5})\) bound holds on a word RAM
with \(O(\log(N/\varepsilon))\)-bit words when each edit cost is an
integer that fits in one word.  This assumes constant-time metric
queries and a rational accuracy parameter occupying \(O(1)\) words;
see Appendix~\ref{sec:finite-precision} for details.

The number of operations is independent of the
alphabet size and the numerical range of \(w\).
In particular, we neither round the
metric to integers nor assume that the input distance exceeds a
specified threshold.
In this version of the manuscript, the theorem concerns a numerical estimate of edit distance;
it does not require the algorithm to output an edit sequence.

\paragraph{Tools.}
We first develop an algorithm that takes an $O(1)$-approximation for the edit distance
as an additional input.  For this task, we build on the dense and sparse
sampling framework~\cite{CDGKS20}, which distinguishes source substrings
with many inexpensive target comparisons from those with few, and on
Andoni's use of sampled target substrings as centers in a hierarchical
search~\cite{Andoni18,Andoni20}.
We extend these ideas to weighted costs using \emph{gap mass}, the
total deletion cost of a substring, to localize comparisons and
partition source blocks at geometrically spaced mass levels.
We select a small family of target substrings to compare with these
pieces, keeping its size comparable to that in the unit-cost setting.
We combine this construction with Kuszmaul's removal of inexpensive
characters~\cite{Kuszmaul19} and Klein's data structure for
shortest-path queries in planar graphs~\cite{Klein05}.
Finally, following Kuszmaul~\cite{Kuszmaul19}, we show how to choose
suitable distance guesses without making the running time depend on
the aspect ratio of the costs.

\paragraph{Further related work.}
Sublinear-time algorithms inspect only part of the input and
distinguish instances with edit distances below a smaller threshold from those above a larger
threshold.  This line of work relates the separation between the
thresholds to the number of inspected characters and the running
time~\cite{BEKMRRS03,GKS19,KociumakaSaha20,BCFN22Gap,GKKS22,BCFK24}.
Bounded-distance algorithms exploit a small upper bound on the
distance to accelerate exact computation, both for unit
costs~\cite{Ukkonen85,Mye86,LV88,LMS98,CGK16Streaming} and for weighted
costs of at least one~\cite{DGHKS23,CKW23,GK25}.  Parallel algorithms also exploit
small distance to reduce the total work while distributing it among
processors~\cite{DDGLS23}.
Quantum algorithms study how quantum access to the input can reduce
the time needed for distance computation and
approximation~\cite{BEGHS18,GJKT24}; conditional lower bounds constrain
the achievable speedups~\cite{BPS21}.
Pattern matching with edits asks for the positions in a text where
a given pattern matches a substring within a prescribed edit
budget~\cite{Sellers80,LV89,GalilPark90,SV96,Myers99,ColeHariharan02,CKW20,CKW22}.
This direction also includes approximation~\cite{CDK19} and weighted variants~\cite{CKW25}.

Other models study how much information about the strings must be
retained or updated.  Streaming algorithms process characters as
they arrive while using limited memory~\cite{CGK16,CGK16Streaming,KPS21},
whereas sketching algorithms
summarize strings separately and subsequently estimate their distance
from the summaries.  These questions connect edit distance to
\emph{document exchange}, where a short message enables one party
to reconstruct another party's similar string~\cite{BZ16}.
Preprocessing asks whether preparing strings in advance can
accelerate subsequent comparisons~\cite{GRS20,BCFN22}.
Incremental and decremental algorithms update distances as strings
are extended or shortened at their ends~\cite{LMS98,KimPark04}.
Fully dynamic algorithms allow character insertions, deletions, and
substitutions at arbitrary positions, with the time per update as a
central measure of complexity~\cite{CKM20,KMS23,GK25,BGK25}.

\paragraph{Independent parallel work.}
Mader, Tavasoli, and Wang~\cite{MaderTavasoliWang26}
independently obtained a randomized \((3+\varepsilon)\)-approximation for metric
weighted edit distance in \(\Ohtilde(N^{7/4}/\varepsilon^8)\)
time.
Their manuscript appeared on arXiv a few days before ours.
Both approaches use the scheme of~\cite{CDGKS20}; theirs builds on \cite{GRS20} and ours follows~\cite{Andoni20}.  Our improved running time
comes from additional optimizations in choosing which substrings to
compare and reusing computations across comparisons.
These optimizations are explained in
\cref{sec:overview-cover,sec:overview-windows,sec:overview-oracles}
and formalized in
\cref{lem:alignment-rectangle-cover,lem:short-pieces-shared-windows,lem:batched-radius-test}.
By mutual agreement, subsequent versions of both manuscripts are intended to be merged.

\section{Technical overview}
\label{sec:technical-overview}

We first develop an algorithm with a supplied distance budget \(D>0\).
For every \(D\), it runs within the stated time and returns an upper
bound on \(\WED_w(X,Y)\), with the desired approximation guarantee when
\(\WED_w(X,Y)\le D<2\WED_w(X,Y)\).
The algorithm for this budget has three parts:
cover all alignments of cost at most \(D\) by rectangles, approximate
edit costs inside each rectangle, and combine these approximations
by a shortest-path computation.  \Cref{sec:overview-cover} describes
the rectangle cover, \crefrange{sec:overview-local}{sec:overview-sparse}
explain the approximation within each rectangle, whereas \cref{sec:overview-shortcuts} discusses
the global shortest-path computation.
\Cref{sec:overview-parameters} fixes the values of various parameters of our algorithm, and \cref{sec:overview-budget} explains how to remove the need for the supplied estimate $D$.

We use an error parameter \(\zeta=\Theta(\varepsilon)\), with a sufficiently small hidden constant.
In this overview, \(O_\zeta\) and \(\Theta_\zeta\) allow their
constants to depend on \(\zeta\); tildes additionally suppress
factors polynomial in \(\log N\).  The formal arguments in \crefrange{sec:model}{sec:aspect-reduction} make the accuracy dependence explicit.

\subsection{Dynamic programming and certified shortcuts}
\label{sec:overview-shortcuts}

For a string $Z$ and integers $0\le p \le q \le |Z|$, sometimes called boundaries,
we denote by \(Z[p\dd q)\) the \emph{fragment} of $Z$ that contains the characters $Z[p],\ldots,Z[q-1]$ at
positions \(p,\ldots,q-1\).
An alignment between $X$ and $Y$ is a path in the usual edit-distance grid
\([0\dd |X|]\times[0\dd |Y|]\), whose vertices are pairs of
boundaries.  An alignment visits \((i,j)\) once it consumes
the first \(i\) characters of \(X\) and the first \(j\) characters
of \(Y\), that is, it aligns prefixes \(X[0\dd i)\) and \(Y[0\dd j)\).
Each grid edge represents an insertion, deletion, substitution, or
match and has the corresponding edit cost.  The standard dynamic program computes \(\WED_w(X,Y)\) as the minimum cost of a path from
\((0,0)\) to \((|X|,|Y|)\).

A value \(z\ge\WED_w(X[a\dd b),Y[p\dd q))\) certifies that \(X[a\dd b)\) can be transformed into \(Y[p\dd q)\) at cost at most~\(z\); we represent this transformation by a \emph{certified shortcut} from \((a,p)\) to \((b,q)\) of cost \(z\) in the alignment graph.
Our \emph{shortcut graph} has the same vertex set as the edit grid.
Its edges consist of the certified shortcuts, which we store explicitly,
and all insertion and deletion edges at their original costs,
which we represent implicitly.
Every path from \((0,0)\) to \((|X|,|Y|)\) in this graph describes a
valid transformation from \(X\) to \(Y\), so its cost is an upper bound
on the true distance.
As proved in \cref{sec:weighted-shortcut-geometry}, the cheapest such
path can be found in time nearly linear in the total input length and
the number of shortcuts, without constructing the full edit grid.

\subsection{Covering the alignment grid}
\label{sec:overview-cover}

For unit costs, every alignment of cost at most \(D\) stays in the
band \(\lvert i-j\rvert\le D\) around the main diagonal.
We seek an analogous restriction for weighted costs and cover the
region that inexpensive alignments can visit with rectangles, so that
we can construct shortcuts separately within each rectangle.
We measure progress along each string by the
total deletion cost of the consumed prefix.

The \emph{gap mass} of a character is \(g(a)=w(a,\gap)\), and
\(\mu(Z)\) denotes the sum of these masses in a string \(Z\).
For a boundary \(i\) of \(Z\), write \(\mu_Z(i):=\mu(Z[0\dd i))\)
for its prefix mass.  The reverse triangle inequality gives
\(\abs{g(a)-g(b)}\le w(a,b)\).  Thus, any alignment (path in the edit-distance grid)
from \((a,p)\) to \((b,q)\), of cost \(c\), satisfies
\[
 \abs{(\mu_X(b)-\mu_X(a))-(\mu_Y(q)-\mu_Y(p))}\le c.
\]
In particular, every alignment of cost at most \(D\) is confined to
vertices $(i,j)$ with \(\abs{\mu_X(i)-\mu_Y(j)}\le D\).

The cover construction, provided in \cref{sec:alignment-cover}, returns a collection of rectangles such that every alignment of cost at most \(D\) lies within their union. Each rectangle has two projections: a source projection corresponding to a fragment of \(X\), and a target projection corresponding to a fragment of \(Y\). Every rectangle is either a unit rectangle, whose two projections each contain one character, or has total gap mass \(O(D)\) in both projections.
The source projections are disjoint, while each character of \(Y\) belongs to the target projection of only \(O(1)\) rectangles.
Consequently, the sum of the lengths of all projections is \(O(N)\), allowing us to sum the running times of the local computations.

The construction first isolates characters of gap mass greater than
\(3D\).  In an alignment of cost at most \(D\), each such character
must be paired with a character of gap mass greater than \(2D\)
whose starting prefix mass differs by at most \(D\).
There is at most one such partner, and we cover these forced pairs by unit rectangles.
Within each remaining pair of fragments, we greedily end a source
piece as soon as its mass reaches \(D\), merging a final piece of
smaller mass with its predecessor when one exists.
For each source piece $S$ in $X$, the prefix-mass bound identifies a target
fragment $T$ in $Y$ of mass \(O(D)\) so that every global alignment of cost at most $D$ aligns $S$ with a fragment of $Y$ contained in $T$; the spacing between source
cuts ensures constant overlap of the target intervals.
If a cover check rules out an alignment of cost at most \(D\),
we return an empty collection of rectangles.

\subsection{The local problem within rectangles}\label{sec:overview-local}
We model each non-unit rectangle as a pair of strings \(S\) and \(T\), with total length
\(m=|S|+|T|\le N\) and total mass \(F=\mu(S)+\mu(T)\).
The local problem we solve within the rectangle is to construct
shortcuts in the edit grid of $S$ and $T$ in \(\widetilde O_\zeta(mN^{3/5})\) time.
With high probability, the following guarantee holds simultaneously
for every target fragment \(V=T[p\dd q)\): the shortcut graph defined by these shortcuts contains a
path from \((0,p)\) to \((|S|,q)\) of cost at most
\[ (3+\varepsilon)\WED_w(S,V)+\tfrac{\varepsilon mF}{N}.\]
The algorithm for the local problem is developed in \crefrange{sec:short-pieces-shared-windows}{sec:rectangle-certification}, culminating in \cref{thm:rectangle-certification}.
\Crefrange{sec:overview-windows}{sec:overview-sparse} provide an overview of this algorithm.

The global algorithm, described in \cref{sec:approximation-algorithm},
combines the local shortcuts and the exact diagonal of each unit
rectangle, then computes a shortest path in their joint shortcut graph.
Its error analysis uses the guarantee for the
portions of an optimal alignment inside each rectangle.  The replacements preserve their
endpoints, and intervening gap edges remain available.  Since
\(F=O(D)\) in every non-unit rectangle and the sizes sum to
\(O(N)\), the additive errors sum to \(O(\varepsilon D)\) and the running times sum to \(\widetilde O_\zeta(N^{8/5})\).

\subsection{Reducing to row--window comparisons}
\label{sec:overview-windows}
Our local algorithm follows the high-level approach of \cite{CDGKS20}, with further simplifications in~\cite{Andoni20}, and multiple non-trivial adaptations to the weighted setting.
We partition the source string \(S\) into \(R\) consecutive fragments \(U_0,\ldots,U_{R-1}\), called \emph{rows}. Each row contains at most \(W\) characters, where $W$ is a parameter set in \cref{sec:overview-parameters}, and the number of rows satisfies
$$
R=\widetilde O_\zeta(m/W).
$$
We also construct a family \(\mathcal C\) of target windows, where every window is a fragment of the target string \(T\). Writing \(M:=|\mathcal C|\) for the number of windows, we have
$$
M=\widetilde O_\zeta(N/W).
$$
The row partition and window family are constructed in \(\widetilde O_\zeta(m)\) time.

To formulate the guarantees these row and window families satisfy, consider an \emph{ideal shortcut graph} obtained by adding a shortcut of cost exactly
\(\WED_w(U,V)\) for every row \(U\) and every window
\(V\in\mathcal C\).
Namely, the construction, provided in \cref{sec:short-pieces-shared-windows} and sketched below, guarantees that any alignment of \(S\) with a fragment of \(T\)
with cost \(c\) can be replaced by a path in the ideal shortcut graph
with the same entry and exit vertices and cost at most
$$
c+O\!\left(\zeta\left(c+\frac{mF}{N}\right)\right).
$$
Thus, the window family suffices for all the target fragments that may be aligned with the different source rows, even though these fragments are not known in advance.

\paragraph{Construction of rows and windows.}
We first split \(S\) into blocks of at most \(W\) characters and
subdivide them using mass levels \(\Delta,2\Delta,4\Delta,\ldots\)
up to \(\mu(S)\), where \(\Delta=\widetilde \Theta_\zeta(FW/N)\) is chosen sufficiently
small.  Since \(\mu(S)\le F\), there are only
\(O_\zeta(\log N)\) mass levels, independently of the range of character masses.
The construction proceeds as follows.
While a block's
unassigned suffix has mass at least \(\Delta\), choose the largest
level \(\tau\) not exceeding the mass, and take the shortest prefix
of mass at least \(\tau\) as the next row.
The mass left over is less than \(\tau\), so successive rows use
strictly decreasing levels; any nonempty remainder of mass below
\(\Delta\) becomes the final row.  This gives the stated row count
and gives enough budget to pay \(O(\Delta)\) additive error per row.

Next, we choose which target boundaries may start a window in $T$.
For this, it suffices to pick any inclusion-wise minimal set of boundaries such that,
for any boundary, we can move forward to a
retained boundary while skipping characters of total mass less than
\(\Delta\).  There are \(O(F/\Delta)=\widetilde O_\zeta(N/W)\) retained starts.
From each retained start, we choose candidate window
masses near every source mass level \(\tau\), using fine spacing near
\(\tau\) and progressively coarser spacing farther away.
Each candidate mass is measured from the chosen start.
If it falls inside a character, we retain the character's ending
boundary; if it falls at a boundary, retain that boundary.
Each start contributes \(\widetilde O_\zeta(1)\) windows over
all levels, giving the stated bound on \(M\).

To see why these windows suffice, consider an alignment of a row
\(U\) with a target fragment \(V\), of cost \(c\).  The mass-imbalance
bound gives \(\mu(U)+\mu(V)-c\le2\min\{\mu(U),\mu(V)\}\).
Thus, if either mass is \(O(\Delta)\), deleting \(U\) and inserting
\(V\) incurs only \(O(\Delta)\) additive loss.
Otherwise, move the target start forward as above, and let \(x\)
be the remaining target mass.  The candidate spacing lets us round
its end backward while removing mass at most \(\zeta|x-\tau|+\Delta\).
If \(|x-\tau|=O(c+\Delta)\), using this window as a shortcut and
inserting the omitted ends incurs \(O(\zeta c+\Delta)\) additive loss.

Otherwise, the mass-imbalance bound forces \(x-\tau\) to be positive
and much larger than \(c+\Delta\).  Since the row $U$ stops as soon as
its mass reaches \(\tau\), its last character has gap mass greater
than \(x-\tau-c>c\).  This character cannot be deleted and must be paired
with a heavy target character.  At the partner's ending boundary,
the allowed rounding loss is smaller than its gap mass, so rounding
cannot cross the partner and must retain this boundary.
We end the shortcut there and keep the trailing insertions at their
original cost.

\subsection{Approximating the ideal shortcut graph}
\label{sec:overview-costs}
The ideal shortcut graph has very strong approximation guarantees,
but computing all \(R\cdot M\) exact
row--window distances would be too expensive.  Instead, as discussed in \cref{sec:overview-oracles,sec:overview-centers,sec:overview-sparse}, we give every row--window pair a certified
fallback label and compute more accurate direct estimates only for selected pairs.
The resulting shortcut graph still preserves
the relevant paths of the ideal graph up to the required
approximation.

The remaining local task is to replace the exact shortcut costs by certified values while increasing all distances from the left to the right boundary of the edit grid by a factor of at most \(3+O(\zeta)\). This task is implemented independently of how the rows and windows were constructed: it applies to any partition of \(S\) into $R$ rows of at most \(W\) characters and any family of $M$ target windows. For any such choice of rows and windows, \cref{thm:row-comparisons} gives an algorithm that constructs the required certified shortcuts in
$$
 \widetilde O_\zeta\left(
 W^{3/2}\sqrt{mRM}+mW+RW^2+RM\right)
$$
time on every outcome. The construction uses the dense- and sparse-search procedures
described in \cref{sec:overview-centers,sec:overview-sparse}.
We first introduce two comparison oracles used by these procedures.

\subsection{Two comparison oracles}
\label{sec:overview-oracles}
The main difficulty in row--window comparisons is that their lengths may be highly unbalanced: a row has at most \(W\) characters, but a window may be much longer.
To handle this, for a single window \(C\),
we retain \(O_\zeta(W)\) characters with the largest gap masses, in
their original order, obtaining a short subsequence \(R_C\) of \(C\).
Let \(\lambda_C=\mu(C)-\mu(R_C)\) be the omitted mass.  A suitable
choice of the retained count ensures that, for every row \(U\),
\[
 \WED_w(U,C)\le
 Q(U,C):=\WED_w(U,R_C)+\lambda_C
 \le(1+\zeta)\WED_w(U,C).
\]
The same $R_C$ works for every row since the length of every row is at most $W$.
After linear preprocessing of the string \(T\),
we can extract $R_C$ in $\widetilde O_\zeta(W)$ time and compute
\(Q(U,C)\) in \(\widetilde O_\zeta(W^2)\)~time.

Why can the long window be shortened?  If no character is omitted,
the comparison is exact.  Otherwise, let \(t\) be the smallest gap
mass among the retained characters, so every omitted character has
gap mass at most \(t\).  Start with an alignment of \(U\) with \(C\).
If an omitted character of \(C\) was inserted, remove that insertion;
if it was paired with a character of \(U\), replace the pairing by
deleting the corresponding character of \(U\).  This produces an
alignment of \(U\) with \(R_C\).  Adding \(\lambda_C\) accounts for
inserting the omitted characters back.  An originally inserted
character causes no increase, while replacing a pairing increases
the cost by at most \(2t\), by the triangle inequality.  At most
\(W\) characters can be paired with characters of \(U\), so the
total increase is at most \(2Wt\).  On the other hand, any alignment
must insert all but at most \(W\) of the retained characters.
Retaining sufficiently many characters therefore makes this
unavoidable insertion cost pay for the increase.

The preceding construction, formalized in \cref{lem:short-string-approximation}, gives a certified $(1+\zeta)$-approximate estimate for a single chosen row–window pair. Computing such an estimate for every pair, however, would still be too expensive. The algorithm therefore uses a second, cheaper oracle to screen many windows against the same row. Given a row \(U\) and a radius \(r>0\), this gap oracle accepts every window whose distance from \(U\) is below \((1-2\zeta)r\) and rejects every window whose distance is at least \(r\). For distances between these two thresholds, either answer is allowed.

The gap test, implemented in \cref{lem:batched-radius-test}, uses the same shortening idea as the direct comparison, now at a scale determined by \(r\).
For a queried window \(V\), retain only its characters whose gap mass is greater than \(\zeta r/W\). Compare \(U\) with this retained subsequence and add the total gap mass of the omitted characters.
As in the argument above, the resulting value is a certified upper bound on \(\WED_w(U,V)\) and exceeds it by at most \(2\zeta r\). Testing whether this value is below \(r\) therefore gives the stated acceptance and rejection guarantees. A single data structure on top of \(T\) supports this filtering for every \(r>0\), so the radii need not be chosen or discretized in advance.

The gap oracle can share its work when the algorithm tests many overlapping windows.
Let \(I\) denote the range of possible starting positions, with its width \(|I|\)
measured by gap mass.  After filtering, every window that might be
accepted contains only \(O_\zeta(W)\) retained characters; otherwise,
the unavoidable insertion cost would already exceed \(r\).
Moreover, because every retained character has
gap mass greater than \(\zeta r/W\), the possible starting positions
span only \(O_\zeta(W|I|/r)\) retained characters.  Thus, all remaining
comparisons lie in one short portion of the filtered version of $T$, and they can
share a single preprocessing step.  Klein's planar distance
oracle~\cite{Klein05} gives preprocessing time
\[
 \widetilde O_\zeta\left(
 W^2\left(1+\frac{|I|}{r}\right)\right),
\]
after which each window is tested in logarithmic time.  This
preprocessing cost can alternatively be bounded by \(\widetilde O_\zeta(mW)\) since the number of retained characters trivially does not exceed $|T|\le m$.

\subsection{Radii and fallback labels}
\label{sec:overview-centers}

We next define the \emph{sampled-center procedure}, implemented in \cref{sec:target-centers}, which uses the
certified comparison value \(Q\) from the preceding subsection to
reduce the number of row--window pairs that require further
attention.
Following \cite{Andoni20}, our
procedure samples a small collection of target windows, called
\emph{centers}, and uses them to assign a radius \(r(U)\) to every row
\(U\).  For every row \(U\) and target window \(V\), it also assigns a
\emph{fallback label} \(z(U,V)\), which is the certified shortcut
cost used unless a more accurate estimate is found later.
On every outcome, these values satisfy
\[
 0\le r(U)\le \mu(U)
 \qquad\text{and}\qquad
 \WED_w(U,V)
 \le z(U,V)
 \le \WED_w(U,V)+2(1+\zeta)r(U).
\]
The construction takes an integer parameter \(1\le k\le M\), chosen in \cref{sec:overview-parameters},
and guarantees, with high probability, simultaneously for every row $U$, that at most \(k\) target windows $V$ satisfy \(\WED_w(U,V)<r(U)\).

The purpose of the radius is to identify the comparisons for which
the fallback label may be insufficient.  If
\(\WED_w(U,V)\ge r(U)\), then
\[
 z(U,V)
 \le \WED_w(U,V)+2(1+\zeta)r(U)
 \le (3+2\zeta)\WED_w(U,V).
\]
Thus, the fallback label is already sufficiently accurate for every
window at distance at least \(r(U)\) from $U$.  Only windows closer than
\(r(U)\) may require a better label, and with high probability there
are at most \(k\) such windows per row.

The procedure samples
\(\widetilde O_\zeta(M/k)\) centers uniformly among target windows.  It also
includes the empty string as a virtual center.  For every row \(U\)
and every center \(C\), compute \(Q(U,C)\).  Let \(C_U\)
be the center minimizing this value, and set
\[
 r(U):=\frac{Q(U,C_U)}{1+\zeta}.
\]
The virtual center ensures \(r(U)\le\mu(U)\) since transforming
\(U\) into the empty string costs \(\mu(U)\).

Why are only \(k\) windows closer than \(r(U)\)?  With high
probability, the sample contains one of the \(k\) windows closest to
each row.  Fix such a sampled window \(C\) for \(U\).  Since \(C_U\)
minimizes \(Q\),
\[
 r(U)
 \le \frac{Q(U,C)}{1+\zeta}
 \le \WED_w(U,C).
\]
Consequently, every window at distance less than \(r(U)\) must appear
before \(C\) among the windows ordered by distance from \(U\).
There can therefore be at most \(k\) such windows.

It remains to construct the fallback labels.  Let \(R_{C_U}\) be the short representation of the chosen
center constructed in the preceding subsection, and define
\[
 z(U,V):=
 \WED_w(U,R_{C_U})+\WED_w(R_{C_U},V).
\]
This label is certified because its two terms describe a valid
transformation from \(U\) to \(V\) through \(R_{C_U}\).  Moreover,
symmetry and the triangle inequality give
\begin{align*}
 z(U,V)
 & \le \WED_w(U,R_{C_U})+\WED_w(R_{C_U},U)+\WED_w(U,V)\\
 &= \WED_w(U,V)+2\WED_w(U,R_{C_U})\\
 &\le \WED_w(U,V)+2(1+\zeta)r(U).
\end{align*}

For each center \(C\), dynamic programming computes \(Q(U,C)\)
in \(O_\zeta(|U|W)\) time.  Since the rows partition \(S\),
all rows together cost
\(O_\zeta(W\sum_U|U|)=O_\zeta(|S|W)\) per center.
Over the \(\widetilde O_\zeta(M/k)\) sampled centers, these
comparisons take \(\widetilde O_\zeta(mWM/k)\) time.
For each chosen center \(C\), one
planar distance oracle~\cite{Klein05}, constructed in $\widetilde O_\zeta(mW)$ time, computes the distances from \(R_C\) to all target
windows, and these answers are reused by every row that chooses
\(C\).
The sampled-center procedure therefore takes
\[
 \widetilde O_\zeta\left(\frac{mWM}{k}+RM\right)
\]
time.  It gives a certified fallback label for every row--window
pair while leaving at most \(k\) close windows per row for the search
procedure described next.

\subsection{A sparse search for the close windows}
\label{sec:overview-sparse}

The sampled-center procedure leaves, for each row, at most \(k\) windows that may
need a better label, but it does not identify these
windows.  Preprocessing a full-target gap oracle for every row would
still be too expensive.  The sparse-search procedure, developed in \cref{sec:sparse-hierarchy}, locates the relevant close windows by
using their order along an alignment: as the alignment moves from one
source row to another, the corresponding target position changes
predictably unless the alignment pays a substantial cost between the
two rows.

To exploit this relation, we follow~\cite{Andoni20} and organize the consecutive source rows in a
balanced binary tree.  The leaves are the individual rows, and every
internal node represents a consecutive block of rows.  At each node,
sample $\Ohtilde_\zeta(1)$ rows from its block, choosing each row with
probability proportional to its radius.  Rows with radius zero can
be omitted as their fallback labels are exact.

For the rows sampled at the root, use the gap oracle to test every
target window.  An accepted row--window pair is called an
\emph{anchor}; it records the position of a close target window for
that row.  Now consider a sampled row \(U\) in a child node.  An anchor
stored at the parent predicts where a close window for this new row
should begin: moving between the two source rows changes the source
prefix mass by a known amount, so the target prefix mass should
change by approximately the same amount.  If the parent contains
\(t\) rows, the gap oracle searches within prefix-mass distance
\(Lt\,r(U)/\zeta\) of this prediction, where \(L=O(\log N)\)
is the number of tree levels.
The accepted pairs become anchors for the child, and the process
continues down the tree.  At each leaf, every accepted pair $(U,V)$ receives
a shortcut corresponding to the direct certified value \(Q(U,V)\).

Fix a path of cost \(c\) in the ideal shortcut graph.  Split it into row subpaths,
keeping each shortcut whole, then expand the shortcuts into optimal
alignments.  Call a row \(U\)
\emph{good} if its portion costs less than \((1-2\zeta)r(U)\) and
every tree block containing it has path cost at most
\(Lt\,r(U)/\zeta\), where \(t\) is that block's row count.
This set is fixed before sampling.  Bad (non-good) rows have total radius
at most \(c/(1-2\zeta)+\zeta c\): the first condition charges radii
to row costs, and violations of the second add at most \(\zeta c\),
since each row belongs to at most \(L\) blocks.

A \emph{useful anchor} records the path's window for a sampled
good row.  The algorithm retains all anchors, but the proof follows
these useful ones.  If every node on a root-to-leaf path samples a
good row, useful anchors propagate down that path: the prefix-mass
bound limits prediction error by the parent block's path cost,
which the sampled good row at the child bounds by its search radius.
The leaf therefore obtains a direct label for its window on the fixed ideal path.

Consider the first nodes on root-to-leaf paths that sample no good
row.  Their blocks are disjoint.  With high probability, bad
rows carry at least a \(1-\zeta\) fraction of each such block's
total radius; otherwise sampling would find a good row.
Thus, all rows requiring fallback labels, including non-close
comparisons, have total radius at most \((1+O(\zeta))c\).
Keep gap edges unchanged and use direct or fallback labels for the
shortcuts.  The resulting cost is at most
\[
 (1+\zeta)c+2(1+\zeta)(1+O(\zeta))c=(3+O(\zeta))c.
\]
A union bound over fixed shortest ideal paths gives this guarantee
simultaneously for all reachable pairs of left and right endpoints.

For the running-time analysis, recall that every window accepted
for row \(U\) has distance less
than \(r(U)\).  \Cref{sec:overview-centers} guarantees that, with high probability,
there are at most \(k\) such windows, so each sampled row produces
at most \(k\) anchors.  On this event, let \(t\) be the number of rows in
a parent node and consider a sampled child row \(U\).  Around each
parent anchor's prediction, we search an interval of gap-mass width
\(\widetilde O_\zeta(t\cdot r(U))\) and merge overlapping intervals for each
queried row before testing windows.  The batched oracle charges
\(\widetilde O_\zeta(W^2(1+|I|/r(U)))\) for an interval \(I\),
so this choice costs \(\widetilde O_\zeta(tW^2)\) per anchor:
the radius cancels.  Each node samples
\(\widetilde O_\zeta(1)\) rows, each producing at most \(k\)
anchors.  At any depth, the parent blocks are disjoint, and their
row counts sum to at most \(R\); each parent has at most two
children.  Summing across levels therefore gives
\(\widetilde O_\zeta(RkW^2)\) preprocessing below the root.

Full-target preprocessing at the root costs
\(\widetilde O_\zeta(mW)\).  Each row belongs to at most one node per
level, and merging intervals ensures that each row--window pair is
tested at most once at that node.  Thus, each pair is tested at most
\(O(\log N)\) times, giving \(\widetilde O_\zeta(RM)\) total
query time.  Combining these costs gives total running time
\[
 \widetilde O_\zeta\left(RkW^2+mW+RM\right).
\]

To enforce the running-time bound on every outcome, stop the
row-comparison construction if it exceeds its prescribed work
bound and return no shortcuts, leaving only gap edges.
This preserves certification; the cap is not reached on the
high-probability sparsity event.

\subsection{Choosing the parameters}\label{sec:overview-parameters}
To complete the description of the algorithm for a supplied budget $D$, only the running-time parameters remain to be
chosen.  The parameter
\(k\) balances the cost of sampling centers,
\(\widetilde O_\zeta(mWM/k)\), against the cost of the sparse search,
\(\widetilde O_\zeta(RkW^2)\).  Choosing
$$
 k=\Theta_\zeta\left(
 \min\left\{M,\sqrt{\tfrac{mM}{RW}}\right\}\right)
$$
gives the row-comparison time stated in \cref{sec:overview-costs}:
$$
 \widetilde O_\zeta\left(
 W^{3/2}\sqrt{mRM}+mW+RW^2+RM\right).
$$

Substituting \(R=\widetilde O_\zeta(m/W)\) and
\(M=\widetilde O_\zeta(N/W)\) gives
\(\widetilde O_\zeta(m\sqrt{NW}+mW+mN/W^2)\).
Balancing the first and last terms gives
\(W=\Theta_\zeta(N^{1/5})\) and local running time
\(\widetilde O_\zeta(mN^{3/5})\).

Finally, the window construction replaces an alignment of cost \(c\)
by a path in the ideal shortcut graph of cost
$$
 c+O\left(\zeta\left(c+\frac{mF}{N}\right)\right),
$$
and the row-comparison algorithm increases this cost by a factor of
at most \(3+O(\zeta)\).  Choosing
\(\zeta=\Theta(\varepsilon)\) sufficiently small gives the local
guarantee from \cref{sec:overview-local}.  Since the rectangle sizes sum to \(O(N)\) and every non-unit
rectangle has \(F=O(D)\), the total running time is
\(\widetilde O_\varepsilon(N^{8/5})\), and the additive errors sum to
\(O(\varepsilon D)\).  If \(d\le D<2d\), where
\(d:=\WED_w(X,Y)\), this error is \(O(\varepsilon d)\), yielding the
claimed \((3+\varepsilon)\)-approximation.
We next explain how to obtain such a budget $D$.

\subsection{Choosing the distance budget}
\label{sec:overview-budget}

To find a suitable budget, handle empty or equal strings exactly and put
\(d:=\WED_w(X,Y)>0\).  Adapting Kuszmaul's random simplification
and scale search~\cite{Kuszmaul19} to unequal lengths, we find
positive \(L\le U\) with \(U/L=O(N)\) and, with high probability,
\(L\le d\le U\).  This takes \(O(N\log^2 N)\) time independently
of the numerical range of the costs; \cref{sec:aspect-reduction}
gives the details.

Starting at \(L\), double the budget \(D\) until it reaches or exceeds
\(U\).  There are \(O(\log N)\) budgets, and assuming the interval
contains \(d\), one budget satisfies \(d\le D<2d\).  Run the algorithm
for each budget and return the smallest answer.  Every answer is
an upper bound on \(d\), even for an unsuitable budget, while the
good budget gives the desired approximation with high probability.
Thus, the reduction preserves the approximation guarantee and multiplies the running time by $O(\log N)$.

\section{Preliminaries}
\label{sec:model}
\label{sec:preliminaries}

This section fixes the notation and basic properties used throughout
the paper.  After fixing basic conventions, we describe weighted edit
distance through its alignment grid, and then introduce the gap-mass
coordinates used to locate low-cost alignments.

\subsection{Notation and conventions}

\paragraph{Strings and intervals.}
A string \(Z\) is a finite sequence of characters from an alphabet
\(\Sigma\).  Its length is denoted by \(|Z|\), and its characters
are indexed from zero: \(Z=Z[0]\cdots Z[|Z|-1]\).
We write \(\epsilon\) for the empty string.
For integers \(a,b\), we write
\[
 [a\dd b]=\{i\in\mathbb Z:a\le i\le b\};
\]
replacing either square bracket by a parenthesis excludes the
corresponding endpoint.
Intervals written with commas, such as \([a,b]\), are real intervals.
For a finite set \(A\), \(|A|\) denotes its cardinality; for a
bounded real interval \(I\), \(|I|\) denotes its length.

\paragraph{Fragments, substrings, and subsequences.}
A string \(Z\) of length \(\ell\) has character positions
\([0\dd\ell)\) and boundaries \([0\dd\ell]\).
Boundary \(i\) separates the first \(i\) characters from the remaining
\(\ell-i\) characters; boundaries \(0\) and \(\ell\) are the two ends.
For \(0\le p\le q\le\ell\), the \emph{fragment}
\(Z[p\dd q)\) is specified by the string \(Z\) and the boundaries
\(p,q\).  Its underlying string \(Z[p]\cdots Z[q-1]\) is a
\emph{substring} of \(Z\).  A fragment thus records an occurrence,
whereas a substring records only its character sequence: distinct
fragments can have the same underlying substring.
We use \(Z[p\dd q)\) for either the fragment or its underlying
substring, as determined by context.  In particular, every empty
fragment \(Z[p\dd p)\) has underlying substring \(\epsilon\).
A prefix is a fragment starting at boundary zero, and a suffix is
a fragment ending at boundary \(\ell\).
A \emph{subsequence} is obtained by retaining an
arbitrary increasing sequence of character positions; those positions
need not be consecutive.

\paragraph{Graphs and paths.}
In a finite directed graph with nonnegative edge weights, the cost of a
directed path is the sum of its edge weights.  For such a graph \(G\),
\(\operatorname{dist}_G(u,v)\) denotes the minimum cost of a directed
path from \(u\) to \(v\), and is \(+\infty\) if no such path exists.

\paragraph{Input size and asymptotics.}
Let \(X=X[0]\cdots X[n-1]\) and \(Y=Y[0]\cdots Y[m-1]\), and write
\(N=n+m\).  Inputs with \(N<2\) are handled exactly in constant time,
so we henceforth assume that \(N\ge 2\).
The accuracy parameter satisfies \(0<\varepsilon\le1\).
All logarithms are to base two unless stated otherwise.
As in the introduction, \(\Ohtilde\) suppresses factors polynomial
in \(\log(N/\varepsilon)\), or in \(\log N\) when no accuracy
parameter is present.  In the technical sections, every polynomial
dependence on \(1/\varepsilon\) is explicit, and hidden constants
are independent of \(\varepsilon\).

\paragraph{Approximation and randomness.}
For a nonnegative distance \(d\) and \(\alpha\ge1\), an
\(\alpha\)-approximation is an estimate in \([d,\alpha d]\).
Probabilities refer to the algorithm's internal random choices.
An event holds \emph{with high probability} if its probability is
at least \(1-N^{-c}\) for any prescribed constant \(c>0\).
Increasing \(c\) changes only constant factors in the running time.
Local calls use the original input-length bound \(N\) in this
convention.  A guarantee stated to hold on every outcome applies
to every choice of the random bits.

\paragraph{Computation.}
We use the real RAM described in the introduction, with exact real
registers and \(O(\log(N/\varepsilon))\)-bit integer registers
\cite[Section~6.1]{EricksonHoogMiltzow24}.
Real addition, subtraction, multiplication, division, and comparison
take constant time.  Real registers do not support floor, digit
extraction, or conversion to integers.
The floors and ceilings used to choose integer parameters are computed
by binary search over integer candidates, using exact real comparisons.
The searched ranges are polynomially bounded in \(N/\varepsilon\), so
each search takes \(\Ohtilde(1)\) time.
A metric query returns the exact cost of a supplied pair of symbols
in constant time.  Generating each independent unbiased random bit
also takes constant time.  Our running times count these operations
and exclude the bit cost of representing and manipulating arbitrary
real numbers.

\subsection{Weighted edit distance and alignments}
\label{sec:weighted-alignments}

Let \(\gap\notin\Sigma\) be the gap symbol, and let \(w\) be a metric
on \(\Sigma\cup\{\gap\}\).  Thus, \(w\) is nonnegative and symmetric,
satisfies the triangle inequality, and has \(w(a,b)=0\) if and only if \(a=b\).
For $a,b\in \Sigma$, substituting \(a\) by \(b\) costs \(w(a,b)\), while inserting or
deleting \(a\) costs
\[
 g(a):=w(\gap,a)=w(a,\gap)>0.
\]
We call \(g(a)\) the \emph{gap mass} of \(a\).

For strings \(U\) and \(V\), their weighted edit distance \(\WED_w(U,V)\)
is the minimum total cost of an edit sequence transforming \(U\)
into \(V\).  We call the first string the \emph{source} and the
second the \emph{target}.

The weighted edit-distance dynamic program~\cite{WagnerFischer74}
has the following graph interpretation.
The weighted alignment grid, also called the edit grid, of \(U\) and
\(V\) has vertex set \([0\dd |U|]\times[0\dd |V|]\).  Vertex \((i,j)\)
represents the prefixes \(U[0\dd i)\) and \(V[0\dd j)\).
For every applicable \(i,j\), the graph contains the
following directed edges:
\[
\begin{array}{rcll}
 (i,j)&\longrightarrow&(i+1,j)
     &\text{of cost \(g(U[i])\),}\\
 (i,j)&\longrightarrow&(i,j+1)
     &\text{of cost \(g(V[j])\),}\\
 (i,j)&\longrightarrow&(i+1,j+1)
     &\text{of cost \(w(U[i],V[j])\).}
\end{array}
\]
These edges represent, respectively, deleting \(U[i]\), inserting
\(V[j]\), and matching or substituting \(U[i]\) with \(V[j]\).
Insertion and deletion edges are called \emph{gap edges}; substitution
and match edges are called \emph{diagonal edges}.

An \emph{alignment} of \(U\) and \(V\) is a directed path from
\((0,0)\) to \((|U|,|V|)\), and its cost is the sum of its edge costs.
More generally, a path from \((a,p)\) to \((b,q)\) represents an
alignment of \(U[a\dd b)\) with \(V[p\dd q)\).  The minimum cost of such a
path is
\[
 \WED_w\bigl(U[a\dd b),V[p\dd q)\bigr).
\]
Because \(w\) is a metric, a sequence of edits applied successively to
the same character can be replaced by a single edit of no greater
cost.  Hence, this alignment formulation is equivalent to the usual
definition of weighted edit distance by edit sequences.

\begin{lemma}[String metric; Sellers~{\cite[Theorem~1]{Sellers74}}]
\label{lem:weighted-string-metric}
The function \(\WED_w\) is a metric on finite strings.  In particular,
for all strings \(R,S,T\),
\[
 \WED_w(R,T)
 \le
 \WED_w(R,S)+\WED_w(S,T).
\]
\end{lemma}

\subsection{Prefix-mass coordinates}
\label{sec:prefix-mass-coordinates}

The mass of a string \(Z\) is
\[
 \mu(Z):=\sum_{j=0}^{|Z|-1}g(Z[j]).
\]
This is the cost of deleting \(Z\), so
\(\WED_w(Z,\epsilon)=\mu(Z)\).
For a boundary \(j\) of \(Z\), define the prefix mass
\[
 \mu_Z(j):=\mu\bigl(Z[0\dd j)\bigr).
\]
Thus, \(\mu_Z(0)=0\), \(\mu_Z(\abs Z)=\mu(Z)\), and
\(\mu\bigl(Z[p\dd q)\bigr)=\mu_Z(q)-\mu_Z(p)\).
Since every character has positive gap mass, the sequence
\(\mu_Z(0),\ldots,\mu_Z(\abs Z)\) is strictly increasing.

The difference between source and target prefix masses is the weighted
analogue of the diagonal coordinate \(i-j\).  The next lemma bounds
its change along any path in an edit grid.

\begin{lemma}[Alignment potential]
\label{lem:potential}
Let \(U\) and \(V\) be strings.  Every path of cost \(c\) from
\((a,p)\) to \((b,q)\) in their edit grid satisfies
\[
 \left|
   (\mu_U(b)-\mu_U(a))-(\mu_V(q)-\mu_V(p))
 \right|
 \le c.
\]
In particular, for every \(D\ge0\), each alignment of \(U\) and \(V\)
of cost at most \(D\) is contained in the strip
\[
 \bigl\{(i,j):|\mu_U(i)-\mu_V(j)|\le D\bigr\}.
\]
\end{lemma}

\begin{proof}
Track the potential \(\mu_U(i)-\mu_V(j)\) along the path.  A deletion or an
insertion changes this potential in absolute value by exactly the cost
of the corresponding edge.  A diagonal edge changes it by
\(g(U[i])-g(V[j])\).  By the reverse triangle inequality for \(w\),
\[
 \bigl|g(U[i])-g(V[j])\bigr|
 =
 \bigl|w(U[i],\gap)-w(V[j],\gap)\bigr|
 \le w(U[i],V[j]),
\]
which is the cost of that edge.  Therefore, the absolute change in the
potential along the entire path is at most the sum of its edge costs.
For an alignment of cost at most \(D\), apply this bound to each
prefix and use \(\mu_U(0)=\mu_V(0)=0\) to obtain the strip containment.
\end{proof}

\subsection{Threshold subsequences and local instances}
\label{sec:threshold-subsequences}
\label{par:threshold-subsequences}

For a string \(Z\) and a threshold \(t>0\), let \(Z^{>t}\) denote the
subsequence obtained by retaining exactly the characters \(Z[j]\) with
\(g(Z[j])>t\), in their original order.  Thus, passing from \(Z\) to
\(Z^{>t}\) removes every character of gap mass at most \(t\).
Unless stated otherwise, fragment boundaries always refer to the
original, unfiltered string.

Fix an accuracy parameter \(0<\varepsilon\le 1\).  Throughout the
local constructions, we also fix
\[
 \zeta:=c_\zeta\varepsilon,
 \qquad 0<\zeta\le \tfrac1{16},
\]
where \(c_\zeta>0\) is a sufficiently small absolute constant.

The local constructions in
\cref{sec:short-pieces-shared-windows,sec:row-comparisons,sec:rectangle-certification}
receive two strings \(S,T\) and an ambient length bound
\[
 N\ge \abs S+\abs T,
 \qquad N\ge 2.
\]
When a local construction is applied to substrings of the original
input, the parameter \(N\) retains its original global value.  In
particular, \(N\), rather than \(\abs S+\abs T\), determines the
failure-probability scale and the parameter choices used by the local
algorithm.

\section{Computing a path through certified shortcuts}
\label{sec:weighted-shortcut-geometry}

We represent certified substring comparisons by shortcuts in the
edit grid.  The following proposition is a direct weighted extension
of the shortest-path algorithm in~\cite[Section~4]{CDGKS20}.
It keeps all gap edges implicit.
We call a shortcut \((a,p)\to(b,q)\), with \(a<b\) and
\(p\le q\), \emph{certified} if its label \(z\), used as its edge
cost, satisfies
\[
 z\ge\WED_w(X[a\dd b),Y[p\dd q)).
\]

The optimization reduces to a maximum-weight chain problem in two
coordinates.  The proof below follows the reduction from costs to
benefits and the sweep in~\cite[Section~4]{CDGKS20}, replacing
character-coordinate differences by prefix-mass differences.

\begin{proposition}[Shortest paths through certified shortcuts]
\label{prop:implicit-weighted-shortcut-consumer}
Let \(G\) consist of all vertices of the edit grid of \(X\) and \(Y\),
all gap edges at their weighted gap costs, and the following \(M\) shortcuts:
\[
 e:(a_e,p_e)\to(b_e,q_e),\qquad
 a_e<b_e,\quad p_e\le q_e,
\]
with labels \(z_e\ge0\).
The distance from the source \((0,0)\) to the sink \((|X|,|Y|)\) in \(G\)
can be computed in \(O(N+M\log N)\) time and \(O(N+M)\)
space, with all gap edges implicit.  If every label is a
certified upper bound on its substring distance, the result is a
certified upper bound on \(\WED_w(X,Y)\).
\end{proposition}

\begin{proof}
For a shortcut \(e\), define its benefit by
\[
 \operatorname{ben}(e):=
 (\mu_X(b_e)-\mu_X(a_e))+(\mu_Y(q_e)-\mu_Y(p_e))-z_e.
\]
The source--sink distance in \(G\) equals
\begin{equation}
\label{eq:weighted-shortcut-chain-value}
 \mu(X)+\mu(Y)-
 \max_{\substack{e_1,\ldots,e_k\\
                  b_{e_i}\le a_{e_{i+1}},\;
                  q_{e_i}\le p_{e_{i+1}}}}
 \sum_{i=1}^k\operatorname{ben}(e_i),
\end{equation}
where the maximum ranges over chains of shortcuts and includes the
empty chain with value zero.
An all-gap path has cost \(\mu(X)+\mu(Y)\).  Traversing a shortcut's
rectangle using gap edges costs
\((\mu_X(b_e)-\mu_X(a_e))+(\mu_Y(q_e)-\mu_Y(p_e))\), so the shortcut saves
exactly \(\operatorname{ben}(e)\).  A sequence of shortcuts can
occur on a directed path precisely when their endpoints satisfy the
two displayed compatibility inequalities: gap edges fill every
coordinate gap between them.  Summing their savings proves
\eqref{eq:weighted-shortcut-chain-value}.  Replacing a negative-benefit
shortcut by gap edges preserves reachability and reduces cost.

To find the maximum, bucket the shortcut starts and ends by their
source coordinates and sweep \(x\) through \([0\dd |X|]\) in increasing
order.  A prefix-maximum
data structure on target coordinates stores the best benefit of a
chain whose last shortcut ends at an already processed vertex.
The query \(\operatorname{prefixmax}(p)\) returns the largest stored
value at a target coordinate at most \(p\).  A point-maximum update
stores the larger of the current and supplied values at its coordinate.
At coordinate \(x\), first activate all shortcuts ending there.
Then, for each shortcut \(e\) starting there, compute
\[
 F_e:=\operatorname{ben}(e)+\operatorname{prefixmax}(p_e),
\]
and schedule a point-maximum update with value \(F_e\) at coordinate
\(q_e\) when the sweep reaches \(b_e\).  Since \(a_e<b_e\),
this value is computed before it is activated.  Processing ends before
starts permits consecutive shortcuts to share a source boundary.
Initialize all maxima to zero; after the final sweep event, the prefix
maximum at \(|Y|\) is the desired chain value.

A segment tree implements each query and update in
\(O(\log N)\) time and uses \(O(N)\) space.  Prefix masses and
event buckets take \(O(N+M)\) time and space.  Finally, replacing
each certified shortcut by an alignment of its substring pair turns
any source--sink path into an alignment of no greater cost.  This proves
certification of the returned value.
\end{proof}

\section{Covering inexpensive alignments}
\label{sec:alignment-cover}

This section constructs a collection of rectangles whose union contains
every alignment of cost at most a given distance budget \(D>0\).
The rectangles have limited overlap and, except for unit rectangles,
gap mass \(O(D)\) in each projection.  These properties allow us to bound
the total error and running time of the subsequent local computations.

For fragments \(X[a\dd b)\) and \(Y[\ell\dd r)\), the
alignment-grid rectangle
$[a\dd b]\times[\ell\dd r]$
is the subgraph induced by the vertices
$\{(i,j):a\le i\le b,\ \ell\le j\le r\}$.
Its source and target projections are \(X[a\dd b)\) and \(Y[\ell\dd r)\),
and their gap masses are the two mass dimensions of the rectangle.
A \emph{unit rectangle} is \(1\)-by-\(1\): it spans one character on each side,
regardless of their masses.  We also allow degenerate rectangles whose source or target projection
is empty.  Such a rectangle contains only insertion edges when its
source projection is empty, and only deletion edges when its target
projection is empty.
Containment of an alignment means containment of every edge of its path.

\begin{lemma}[Rectangle cover]
\label{lem:alignment-rectangle-cover}
Given strings \(X,Y\) and a budget \(D>0\), a deterministic algorithm in
\(O(N\log N)\) time returns a collection \(\mathcal B\) of
alignment-grid rectangles such that:
\begin{enumerate}
\item Each rectangle is \(1\)-by-\(1\), or both of its mass
      dimensions are \(O(D)\).
\item Each character of \(X\) belongs to the source projection of
      at most one rectangle.
\item Each character of \(Y\) belongs to the target projection of
      at most \(O(1)\) rectangles.
\item Every $X$-to-$Y$ alignment of cost at most \(D\) is contained in the
      union of the rectangles.
\end{enumerate}
\end{lemma}

\begin{proof}
Whenever a check below certifies \(\WED_w(X,Y)>D\), discard any
rectangles constructed so far and return the empty collection
\(\mathcal B=\varnothing\).  This collection satisfies the first three
properties.  The coverage property is vacuous because no alignment
has cost at most \(D\).

\emph{Isolating heavy characters.}
Consider \(X[i]\) with \(g(X[i])>3D\).  Every alignment of cost at
most \(D\) consumes it by a diagonal with some \(Y[j]\).  The
diagonal cost and \cref{lem:potential} imply
\[
 w(X[i],Y[j])\le D,\qquad g(Y[j])>2D,\qquad |\mu_X(i)-\mu_Y(j)|\le D.
\]
These inequalities give
\(\mu_Y(j)\le\mu_X(i)+D<\mu_Y(j+1)\).
Thus, the partner $Y[j]$ is uniquely determined by the target prefix masses.

For each source character heavier than \(3D\), use binary search
over all target prefix masses to find this candidate partner.
If no such interval exists or the candidate violates any of the
displayed necessary conditions, return the empty collection.
Perform the symmetric search for every target character heavier than \(3D\).
Take the union of the proposed pairs, keeping each pair of positions only once.
If the proposed pairs are not one-to-one and order-preserving,
return the empty collection.

For each proposed pair \((i,j)\), output the unit rectangle
\([i\dd i+1]\times[j\dd j+1]\).  The fragments before, between,
and after the paired positions form corresponding source--target
fragment pairs, which we call \emph{components} in this proof.
Omit components empty on both sides.  For a component empty on
exactly one side, return the empty collection if its total mass
exceeds \(D\); otherwise, output its degenerate rectangle.

Every alignment of cost at most \(D\) uses all proposed diagonals,
so the partner and consistency checks must succeed if such an
alignment exists.  Cutting the alignment immediately before and after these
diagonals leaves a subpath through each component.  A component
empty on exactly one side has only a gap alignment, whose cost
equals its total mass; this justifies the mass check.  Components
empty on both sides contain no edges.  Independently of whether a
cheap alignment exists, every character remaining in a component
has mass at most \(3D\), since each heavier character belongs to a
proposed pair.

\emph{Covering a component nonempty on both sides.}
Fix one such component and write its target fragment as
\(Y[u\dd v)\), where \(u<v\) are boundary indices in the
original string \(Y\).
Greedily cut its source at the first boundary where the mass since
the preceding cut reaches \(D\).  Merge a final remainder of mass
below \(D\) into the preceding piece, if there is one.  Write the
resulting pieces as \(X[a_s\dd a_{s+1})\), and put
\(\alpha_s=\mu_X(a_s)\), using global prefix masses.

For each source piece, return the empty collection if either set
below is empty; otherwise, output \([a_s\dd a_{s+1}]\times[\ell_s\dd r_s]\), where
\begin{align*}
 \ell_s&:=\min\bigl\{j\in[u\dd v]:
       \mu_Y(j)\ge\alpha_s-2D\bigr\},\\
 r_s&:=\max\bigl\{j\in[u\dd v]:
       \mu_Y(j)\le\alpha_{s+1}+2D\bigr\}.
\end{align*}
An alignment of cost at most \(D\) meets both source boundaries
inside the component and, by \cref{lem:potential}, supplies a target
boundary in each set.  Thus, if either set
is empty, no such alignment exists.
When both sets are nonempty,
we also have \(\ell_s\le r_s\):
if \(\ell_s>r_s\), then \(r_s<v\), and the definitions would give
\(\mu_Y(r_s)<\alpha_s-2D\) and
\(\mu_Y(r_s+1)>\alpha_{s+1}+2D\).  These consecutive prefix masses
would differ by more than \(4D\), contradicting
\(g(Y[r_s])\le3D\).  Every output is thus a rectangle.

Each piece has mass in \([D,5D)\), except that a component of total source mass
below \(D\) yields just one smaller piece.  Indeed, a greedy piece
has mass less than \(4D\), and the merged remainder adds less than
\(D\).
The definitions give
\(\mu_Y(\ell_s)\ge\alpha_s-2D\) and
\(\mu_Y(r_s)\le\alpha_{s+1}+2D\).  Thus, the target projection has
mass at most \(\mu(X[a_s\dd a_{s+1}))+4D<9D\).

\emph{Overlap.}
If \(Y[j]\) belongs to the target projection of the rectangle for
piece \(s\), then \(\alpha_s-2D\le\mu_Y(j)<\alpha_s+7D\).
Consecutive source-piece starts in one
component are separated in mass by at least \(D\), unless that
component has only one piece.  Hence, each target character belongs
to only a constant number of rectangles of that component.  The
components and unit rectangles have disjoint target projections
as sets of character positions.  Source projections
are disjoint by construction, proving the first three properties.

\emph{Coverage.}
Fix any alignment of cost at most \(D\).  By the preceding arguments,
all checks succeed.  The alignment uses
all the proposed diagonals and stays within the intervening
components.  At each vertex \((i,j)\) of the alignment inside a component,
with \(i\in[a_s\dd a_{s+1}]\),
\cref{lem:potential} gives
\[
 \alpha_s-D\le \mu_X(i)-D\le \mu_Y(j)\le \mu_X(i)+D
 \le\alpha_{s+1}+D.
\]
The lower inequality implies \(\ell_s\le j\), and the upper
inequality gives \(j\le r_s\).  Thus, the rectangle contains
every edge of the alignment within that source piece, including
insertions on either source boundary.  The unit rectangles and
degenerate component rectangles cover the remaining edges.

\emph{Running time.}
The searches for partners and all consistency checks take
\(O(N\log N)\) time.  The greedy source cuts take linear time,
and binary searches find the target boundaries within the same time bound.
\end{proof}

For a returned rectangle \(\mathcal B_s\), let \(m_s\) be the sum
of the character counts in its two projections.  The projection
bounds alone imply
\[
 \sum_s m_s=O(N).
\]
Since the construction omits rectangles with both dimensions zero,
there are \(O(N)\) rectangles.  Henceforth, unit rectangles are
handled by their exact substitution edges, and rectangles with a
zero dimension require only gap edges.

\section{Rows and windows}
\label{sec:short-pieces-shared-windows}
The rectangle cover of \cref{sec:alignment-cover} reduces the global
problem to constructing shortcuts separately inside each non-unit rectangle.
In this section, we develop the first step of this local construction:
partition the source projection into \emph{rows} and select a
family of fragments from the target projection, called \emph{windows}.
We show that exact row--window distances suffice to approximate the
weighted edit distance from the source to every target fragment.

Fix nonempty strings \(S\) and \(T\), let \(m:=|S|+|T|\) and \(F:=\mu(S)+\mu(T)\), and let
\(N\ge m\) be a given length bound.
Consider the alignment-grid rectangle corresponding to \(S\) and \(T\).
After reindexing its coordinates locally, its vertex set is
\([0\dd |S|]\times[0\dd |T|]\).

The rows are the fragments \(U_i=S[x_i\dd x_{i+1})\) of \(S\)
defined by a partition \(0=x_0<\cdots<x_R=|S|\).
Each window is a possibly empty fragment \(T[p\dd q)\) of the target string \(T\).
For a family \(\mathcal C\) of windows, the \emph{ideal shortcut graph}
contains all gap edges in this grid and, for each row \(U_i\) and each
window \(T[p\dd q)\in\mathcal C\), the shortcut
\((x_i,p)\longrightarrow(x_{i+1},q)\) of cost \(\WED_w(U_i,T[p\dd q))\).
The ideal shortcut graph uses exact distances, which the construction
of rows and windows does not compute.

\begin{restatable}[Construction of rows and windows]{lemma}{sharedwindows}
\label{lem:short-pieces-shared-windows}
Given \(W\in[1\dd m]\), one can construct in
\(\Ohtilde(m/\zeta)\) time a partition
\(0=x_0<\cdots<x_R=|S|\) and a family \(\mathcal C\) of windows
with the following properties:
\begin{enumerate}
\item Every row \(U_i=S[x_i\dd x_{i+1})\) has at most \(W\)
      characters, and the number of rows is
      \[
       R=\Ohtilde\!\left(\tfrac{m}{W}\right).
      \]
\item The number of windows is
      \[
       |\mathcal C|
       =\Ohtilde\!\left(\tfrac{N}{\zeta^2W}\right).
      \]
\item For every fragment \(T[u\dd v)\), the ideal shortcut graph contains a path from \((0,u)\) to
      \((|S|,v)\) of cost at most
      \[
       (1+O(\zeta))\cdot \WED_w(S,T[u\dd v))+O\!\left(\tfrac{\zeta mF}{N}\right).
      \]
\end{enumerate}
\end{restatable}

The construction is deterministic, and the approximation guarantee
holds for every fragment of~\(T\).  We prove
\cref{lem:short-pieces-shared-windows} in stages.
First we partition the source, then
choose target starts, and finally determine target ends.  A local rounding argument gives a
replacement path for each row; concatenating these paths completes
the proof.

\subsection{Partitioning the source}

For the rest of this section, fix the maximum row length \(W\in[1\dd m]\).
Define the logarithmic normalization and mass resolution
\[
 \Lambda:=C_\Lambda\left\lceil
 \log_2\!\left(\tfrac{N}{\zeta}\right)\right\rceil
 \qquad\text{and}\qquad
 \Delta:=\tfrac{\zeta FW}{N\Lambda},
\]
where \(C_\Lambda\) is a sufficiently large positive integer constant.
The ceiling of the logarithm is the least integer \(j\) with
\(2^j\zeta\ge N\), found by repeated doubling.
Further, let
\[
 \mathcal L:=\{2^i\Delta:i\in\mathbb Z_{\ge0},\ 2^i\Delta\le\mu(S)\}
\]
be the \emph{mass levels} used in the construction.
The resolution \(\Delta\) bounds the number of levels even when
individual characters have arbitrarily small positive mass.

We construct each row \(U\) of mass at least \(\Delta\) so that its last character
accounts for any excess beyond the row's mass level \(\tau \in \mathcal L\), a property used by the later
endpoint-rounding argument.
To express this property, let \(U^-\) denote the prefix obtained by
removing the last character of \(U\).

\begin{lemma}[Source partition]
\label{lem:local-source-partition}
In \(\Ohtilde(m)\) time, one can partition the source string \(S\) into \(R\) rows
with the following properties:
\begin{enumerate}
\item The number of rows satisfies
\begin{equation}
\label{eq:local-resolution-charge}
 R=\Ohtilde\!\left(\tfrac{m}{W}\right)\qquad\text{and}\qquad
 R\Delta=O\!\left(\tfrac{\zeta mF}{N}\right).
\end{equation}
\item Every row \(U\) has at most \(W\) characters and either has mass less
than \(\Delta\) or satisfies
\[
 \mu(U^-)<\tau\le\mu(U)
 \qquad\text{for some }\tau\in\mathcal L.
\]
\end{enumerate}
\end{lemma}

\begin{proof}
Cut \(S\) into pieces of \(W\) characters, except possibly a shorter
last piece.  Within each piece, process its nonempty unassigned
suffix \(Z\) as follows.  If \(\mu(Z)<\Delta\), make \(Z\) its last row.
Otherwise, choose the largest \(\tau\in\mathcal L\) with
\(\tau\le\mu(Z)\), make the shortest prefix \(U\) with
\(\mu(U)\ge \tau\) the next row, and continue with the remaining suffix.

Each row stays within an original piece of at most \(W\)
characters.  Maximality of the chosen level \(\tau\) gives
\(\mu(U)\le \mu(Z)<2\tau\), and the shortest-prefix rule gives
\(\mu(U^-)<\tau\le\mu(U)\).  The remaining suffix has mass less
than \(\tau\), so successive chosen levels strictly decrease.
Since \(\mu(S)\le F\) and \(F/\Delta=N\Lambda/(\zeta W) \le N\Lambda/\zeta\),
there are \(O(\log(N\Lambda/\zeta))=O(\Lambda)\) levels.  Each original piece therefore
yields \(O(\Lambda)\) rows, including a possible final row of
mass below \(\Delta\).  Consequently,
\[
 R=O\!\left(\left(1+\tfrac{|S|}{W}\right)\Lambda\right)=\Ohtilde\!\left(\tfrac{m}{W}\right),
 \qquad
 R\Delta
 =O\!\left(\left(1+\tfrac{|S|}{W}\right)\Lambda\tfrac{\zeta FW}{N\Lambda}\right)
 =O\!\left(\tfrac{\zeta mF}{N}\right).
\]
Listing the levels, selecting them by binary search, and scanning
row endpoints takes \(\Ohtilde(m)\) time.
\end{proof}

For the remainder of this section, fix a partition supplied by
\cref{lem:local-source-partition}.

\subsection{Choosing target starts}

Next, we construct a small set of target boundaries to which any target start
can be moved forward at cost less than \(\Delta\).  We call this
set the \emph{start net}.

\begin{lemma}[Start net]
\label{lem:local-start-net}
In \(O(m)\) time, one can construct a set
\(P\subseteq[0\dd |T|]\) of size \[|P|=O\!\left(1+\tfrac{\mu(T)}{\Delta}\right)\] and a map
\(\rho:[0\dd |T|]\to P\) such that, for every \(j\in[0\dd |T|]\),
\[
 j\le\rho(j)\qquad\text{and}\qquad \mu(T[j\dd\rho(j)))<\Delta.
\]
\end{lemma}

\begin{proof}
Partition \([0\dd |T|]\) greedily into consecutive clusters
\(C_t=[s_t\dd u_t]\), starting at \(s_0=0\), with
\[
 u_t:=\max\{u\in[s_t\dd |T|]:\mu(T[s_t\dd u))<\Delta\}.
\]
If \(u_t<|T|\), start the next cluster at \(s_{t+1}=u_t+1\).
Keep the last boundary of every cluster and map each boundary to its
cluster's last boundary. Formally,
\[
 P:=\{u_t:C_t\text{ is a cluster}\}\qquad\text{and}\qquad
 \rho(j):=u_t\quad\text{for }j\in C_t.
\]

For each cluster \(C_t\), we have \(\mu(T[s_t\dd u_t))<\Delta\),
which gives the displacement bound.  Between consecutive clusters,
maximality of \(u_t\) gives \(\mu(T[s_t\dd s_{t+1}))\ge\Delta\).
Hence, there are at most \(1+\mu(T)/\Delta\) clusters; each of them
contributes one boundary.

One scan of the prefix masses constructs the
clusters, the ordered set \(P\), and the map \(\rho\).
\end{proof}

For the remainder of this section, fix \(P\) and \(\rho\) supplied
by \cref{lem:local-start-net}.

\subsection{Choosing target ends}

The next lemma constructs a small family of windows that permits
rounding target ends backward with a controlled loss of mass.
We choose candidate masses more densely near each source mass
level; the permitted loss grows with the mass difference from that
level, up to an additive \(\Delta\).  In the row-alignment proof,
either this difference is controlled by the alignment cost and
\(\Delta\), or the row's last character forces a heavy target partner
whose ending boundary is retained exactly.

\begin{lemma}[Target windows]
\label{lem:shared-target-windows}
In \(\Ohtilde(m/\zeta)\) time, one can construct a family
\(\mathcal C\) of target windows whose starting boundaries belong
to \(P\), with the following properties:
\begin{enumerate}
\item The family has size
\[
 |\mathcal C|=\Ohtilde\!\left(\tfrac{N}{\zeta^2W}\right).
\]
\item For every start \(p\in P\), mass level \(\tau\in\mathcal L\), and target
boundary \(t\in[p\dd |T|]\), there exists a window
\(T[p\dd q)\in\mathcal C\) satisfying
\begin{equation}
\label{eq:target-end-rounding}
 p\le q\le t\qquad\text{and}\qquad
 \mu(T[q\dd t))\le\zeta|\mu(T[p\dd t))-\tau|+\Delta.
\end{equation}
\end{enumerate}
\end{lemma}

\begin{proof}
We start by presenting the construction. Let
\[
 K:=\lceil1/\zeta\rceil\qquad\text{and}
 \qquad J:=\min\{j\in\mathbb Z_{\ge0}:2^j\Delta\ge F\}.
\]
For \(i\in[0\dd J)\), choose equally spaced offsets in
\([2^i\Delta,2^{i+1}\Delta)\):
\[
 \mathcal R:=
 \left\{2^i\Delta\left(1+\tfrac jK\right):
       i\in[0\dd J),\ j\in[0\dd K)\right\}
 \cup\{2^J\Delta\}.
\]
Define the set of prescribed relative masses
\[
 \mathcal S:=
 \{\tau:\tau\in\mathcal L\}
 \cup\{\tau-u,\tau+u:\tau\in\mathcal L,\ u\in\mathcal R\}.
\]
For a start \(p\in P\) and mass offset \(s\in[0,\mu(T[p\dd |T|))]\), let
\(q_p(s)\) be the first target boundary \(q\ge p\) for which
\(T[p\dd q)\) has mass at least \(s\), that is,
\[
 q_p(s):=\min\{q\in[p\dd |T|]:\mu(T[p\dd q))\ge s\}.
\]
If a target boundary \(t\in[p\dd |T|]\) satisfies
\(s\le\mu(T[p\dd t))\), then \(t\) already reaches the prescribed
mass \(s\).  Since \(q_p(s)\) is the first boundary reaching that
mass, we have \(q_p(s)\le t\).

Define
\[
 Q(p):=\{p,|T|\}\cup
 \{q_p(s):s\in\mathcal S\cap[0,\mu(T[p\dd |T|))]\}\]
and
\[
 \mathcal C:=\{T[p\dd q):p\in P,\ q\in Q(p)\}.\]

Next, we analyze the size and the construction time.
By \cref{lem:local-start-net} and the definition of \(\Delta\),
\[
 |P|=O\!\left(1+\tfrac{F}{\Delta}\right)
 =\Ohtilde\!\left(\tfrac{N}{\zeta W}\right).
\]
Also, \(|\mathcal L|=O(\Lambda)\) and
\(J=O(\Lambda)\), because
\(F/\Delta=N\Lambda/(\zeta W)\le N\Lambda/\zeta\).
Hence, \(|\mathcal R|=KJ+1=O(\Lambda/\zeta)\).
It follows that
\[
 |Q(p)|=O\bigl(1+|\mathcal L|\,|\mathcal R|\bigr)
 =O(\Lambda^2/\zeta)
 =\Ohtilde(1/\zeta),
\]
proving the size bound.
We can find each \(q_p(s)\) by binary search in the prefix masses.
Enumerating the levels, offsets, and windows, including sorting and
deduplicating, takes
\(\Ohtilde(m/\zeta)\) time since \(P\subseteq[0\dd |T|]\) implies
\(|P|\le |T|+1\le m\).

It remains to prove the approximation guarantee.
Fix a start \(p\in P\), a mass level \(\tau\in\mathcal L\), and a target boundary \(t\in[p\dd |T|]\),
and put \(x:=\mu(T[p\dd t))\).  We first find
\(s\in\mathcal S\cup\{0\}\)
with \[0\le s\le x\qquad \text{and}\qquad
x-s\le\zeta|x-\tau|+\Delta.\]
For \(x\ge\tau\), set \(s=\tau\) if \(x-\tau<\Delta\);
otherwise, round the excess \(x-\tau\) down to the largest
\(u\in\mathcal R\) with \(u\le x-\tau\), and set \(s=\tau+u\).
If \(x<\tau\), round the deficit \(\tau-x\) up to the smallest
\(u\in\mathcal R\) with \(u\ge\tau-x\), and set
\(s=\max\{0,\tau-u\}\).  Both choices ensure \(0\le s\le x\).

To bound the rounding loss \(x-s\), note that consecutive offsets
in \([2^i\Delta,2^{i+1}\Delta]\) are separated by
\(2^i\Delta/K\), so their ratio is at most \(1+1/K\le1+\zeta\).
Since the offsets range from \(\Delta\) to at least \(F\), rounding
in either direction changes an offset of magnitude at least \(\Delta\)
by at most a \(\zeta\) fraction of its value.  Thus,
\(x-s\le\zeta|x-\tau|\) when \(|x-\tau|\ge\Delta\), whereas
\(x-s<\Delta\) otherwise.  Truncating \(s\) at zero can only
reduce this loss.

Take \(q=q_p(s)\).  Since \(s\in\mathcal S\cup\{0\}\) and
\(q_p(0)=p\), we have \(q\in Q(p)\) and hence \(T[p\dd q)\in\mathcal C\).
The inequalities \(0\le s\le x\) give \(p\le q\le t\).
Moreover, \(\mu(T[q\dd t))\le x-s\), proving
\eqref{eq:target-end-rounding}.
\end{proof}

For the remainder of this section, fix a family \(\mathcal C\)
supplied by \cref{lem:shared-target-windows}.

\subsection{Approximating a row alignment}

We now use the row and window properties to approximate the edit
distances between rows and fragments of $T$.
By the triangle inequality (\cref{lem:weighted-string-metric}),
trimming a target prefix or suffix increases the weighted edit distance
by at most the mass removed.
To use a window as a shortcut, we insert the omitted target prefix
before the shortcut and the omitted suffix afterward.  Their mass
contributes both to the shortcut-cost bound and to the cost of
these insertions.  The proof either bounds both omitted masses or
retains the trailing insertions of an optimal path at their original~cost.

\begin{lemma}[A path across one row]
\label{lem:local-row-replacement}
For every row \(U_i=S[x_i\dd x_{i+1})\) and every fragment \(T[j_-\dd j_+)\), the ideal shortcut graph
contains a path from \((x_i,j_-)\) to \((x_{i+1},j_+)\) of cost
at most
\[
 (1+O(\zeta))\WED_w(U_i,T[j_-\dd j_+))+O(\Delta).
\]
\end{lemma}

\begin{proof}
Fix \(i,j_-,j_+\), and let \(\pi\) be a
minimum-cost path in the edit grid from \((x_i,j_-)\) to
\((x_{i+1},j_+)\).  Its cost is
\(c_i:=\WED_w(U_i,T[j_-\dd j_+))\).

\emph{Small masses.}
If either \(\mu(U_i)\) or \(\mu(T[j_-\dd j_+))\) is less than
\(\Delta\), delete \(U_i\) and insert \(T[j_-\dd j_+)\).
By \cref{lem:potential}, the two masses differ by at most \(c_i\),
so their sum is less than \(c_i+2\Delta\).  Gap edges
therefore give the required path.

Assume henceforth that both masses are at least \(\Delta\).
By \cref{lem:local-source-partition}, choose \(\tau\in\mathcal L\)
such that \(\mu(U_i^-)<\tau\le\mu(U_i)\).
\Cref{lem:local-start-net} gives a start \(p:=\rho(j_-)\in P\) such that \(j_-\le p\) and
\(0\le \mu(T[j_-\dd p))<\Delta\).  Since the target interval has mass
at least \(\Delta\), we also have \(p\le j_+\).  Write
\[
 v:=\mu(T[p\dd j_+))-\tau
\]
for the difference between the remaining target mass and the source
mass level \(\tau\).  By \cref{lem:potential}, we have
\(\left|\mu(T[j_-\dd j_+))-\mu(U_i)\right|\le c_i\).
Moving the start from \(j_-\) to \(p\) removes mass less
than \(\Delta\).  Since \(\mu(U_i)\ge\tau\), we obtain
\[
 v\ge\mu(T[j_-\dd j_+))-\Delta - \tau\ge \mu(U_i)-c_i-\Delta-\tau\ge-c_i-\Delta.
\]

\emph{Case 1. \(v\le4(c_i+\Delta)\).}
Since \(|v|\le4(c_i+\Delta)\), \cref{lem:shared-target-windows} supplies a
window \(T[p\dd q)\in\mathcal C\) with
\[
 p\le q\le j_+\qquad\text{and}\qquad
 \mu(T[q\dd j_+))\le4\zeta(c_i+\Delta)+\Delta.
\]
Insert the prefix \(T[j_-\dd p)\), take the shortcut for
\(T[p\dd q)\), and insert the suffix \(T[q\dd j_+)\).
By the triangle inequality, the shortcut cost is at most \(c_i\)
plus the total mass of the omitted prefix and suffix.
The total path cost is therefore at most
\[
 c_i+2\mu(T[j_-\dd p))+2\mu(T[q\dd j_+))
 \le c_i+2\Delta + 8\zeta(c_i+\Delta)+2\Delta.
\]
This is \(c_i+O(\zeta c_i+\Delta)\), as required.

\emph{Case 2. \(v>4(c_i+\Delta)\).}
Let \(a\) be the last character of \(U_i\).  By \cref{lem:potential},
\[
 \mu(U_i)\ge\mu(T[j_-\dd j_+))-c_i
 \ge\mu(T[p\dd j_+))-c_i=\tau+v-c_i.
\]
Since the prefix \(U_i^-\) preceding \(a\) has mass less than
\(\tau\), the last character has gap mass
\[
 g(a)=\mu(U_i)-\mu(U_i^-)>v-c_i>c_i.
\]
The path \(\pi\) cannot delete \(a\), since that alone would cost
more than \(c_i\).  It therefore pairs \(a\) diagonally with a
target character \(T[h]\), for which
\[
 g(T[h])\ge g(a)-w(a,T[h])>v-2c_i>\Delta.
\]
The displacement bound in \cref{lem:local-start-net} forces
\(p\le h\): if \(p>h\), then \(j_-\le h<p\) gives
\(\mu(T[j_-\dd p))\ge g(T[h])>\Delta\), a contradiction.
Since \(a\) is the last source
character, every edge of \(\pi\) after this diagonal is an
insertion.  Write the cost of these insertions as
\(\beta:=\mu(T[h+1\dd j_+))\le c_i\).  Since
\(\mu(T[p\dd h+1))-\tau=v-\beta>0\), applying
\cref{lem:shared-target-windows} to the boundary \(h+1\) gives a
window \(T[p\dd q)\in\mathcal C\) with \(q\le h+1\) and
\[
 \mu(T[q\dd h+1))\le\zeta(v-\beta)+\Delta
 \le\tfrac{v}{2}+\Delta<v-2c_i<g(T[h]),
\]
where we use \(\zeta\le\tfrac12\) and \(v>4(c_i+\Delta)\).
Since \(g(T[h])=\mu(T[h\dd h+1))\), this forces
\(q=h+1\).

Insert \(T[j_-\dd p)\), use the shortcut for \(T[p\dd h+1)\), and
retain the trailing insertions of \(\pi\).  The prefix of \(\pi\)
ending at \((x_{i+1},h+1)\) costs \(c_i-\beta\).  By the triangle
inequality, deleting \(T[j_-\dd p)\) from its target bounds the
shortcut cost by \(c_i-\beta+\mu(T[j_-\dd p))\).  The initial and final
insertions cost \(\mu(T[j_-\dd p))\) and \(\beta\), respectively.
Thus the total path cost is at most
\(c_i+2\mu(T[j_-\dd p))<c_i+2\Delta\), completing the proof.
\end{proof}

\subsection{Completing the decomposition proof}

\sharedwindows*

\begin{proof}
\Cref{lem:local-source-partition,lem:local-start-net,lem:shared-target-windows}
establish the row- and window-count bounds and the claimed
construction time.
Fix a target fragment \(T[u\dd v)\), and let \(\pi\) be an optimal
alignment of \(S\) with this fragment, of cost
\(c:=\WED_w(S,T[u\dd v))\).
Slice \(\pi\) at the row boundaries, and write \(c_i\) for its row costs.
The edit distance across each row is at most the cost \(c_i\) of
its subpath.  Hence, \cref{lem:local-row-replacement} gives a
replacement path of cost \(c_i+O(\zeta c_i+\Delta)\) for each row $U_i$.
These replacements preserve all slicing vertices, so they can be concatenated.
The total cost of their concatenation is
\[
 c+O(\zeta c+R\Delta)
 =c+O\!\left(\zeta\left(c+\tfrac{mF}{N}\right)\right),
\]
by \eqref{eq:local-resolution-charge}.  This proves the
approximation guarantee.
\end{proof}

\section{Comparison oracles}
\label{sec:comparison-oracles}

This section implements the certified comparisons used by the
sampling algorithms.  A single comparison with a short source string
admits a direct multiplicative approximation.  For batched comparisons,
a planar distance oracle~\cite{Klein05} shares preprocessing across many target
windows and can be restricted to an interval of possible target starts.

Throughout this section, \(N\ge2\) is an upper bound on the length of every
string, and the \(\Ohtilde\) notation suppresses factors polynomial
in \(\log N\).

We repeatedly use an observation that permits removing
any subset of the light target characters.

\begin{observation}[Restoring light target characters]
\label{obs:target-restoration}
For every pair of strings \(U\) and \(V\), threshold \(t>0\), and subsequence \(R\) obtained from \(V\) by removing only characters of gap mass at most \(t\),
\[
 \WED_w(U,V)\le \WED_w(U,R)+\mu(V)-\mu(R)
 \le \WED_w(U,V)+2|U|t.
\]
\end{observation}

\begin{proof}
Restoring the removed characters proves the first inequality.
For the second, project an optimal alignment onto \(R\)
and add \(\mu(V)-\mu(R)\).  A removed insertion keeps its cost; by the triangle inequality, a
removed diagonal \((a,b)\) now costs
\(g(a)+g(b)\le w(a,b)+2g(b)\le w(a,b)+2t\).
There are at most \(|U|\) diagonals, so the total cost increase is at most $2|U|t$.
\end{proof}

\subsection{Approximating one short-string comparison}

For a single comparison, we retain a bounded number of the heaviest
target characters.  Any alignment with a short source must insert
most of these characters.  Their insertion cost pays for the error
from discarding the remaining target characters.

\begin{lemma}[Short representations for direct comparisons]
\label{lem:short-string-approximation}
Every string \(Z\) admits \(O(|Z|+1)\)-time and space preprocessing
such that, given an integer \(W\ge1\), a fragment \(C=Z[p\dd q)\),
and \(0<\varepsilon\le1\), one can in \(O((W/\varepsilon)\log N)\) time construct a string \(R_C\)
of \(O(\min\{|C|,W/\varepsilon\})\) characters and a value \(\lambda_C\ge0\), with the following
properties for every string \(U\) of length \(|U|\le W\):
\begin{enumerate}
\item The value
\[
 Q(U,C):=\WED_w(U,R_C)+\lambda_C
\]
satisfies
\[
 \WED_w(U,C)\le Q(U,C)\le(1+\varepsilon)\WED_w(U,C).
\]
\item Given \(U\) and the constructed \(R_C,\lambda_C\), the value
\(Q(U,C)\) can be computed in \(O(W^2/\varepsilon)\) time.
\end{enumerate}
\end{lemma}

The same representation works for all short source strings.  This
will let the center construction reuse each \(R_C\) as a pattern
in a batched distance oracle.

\begin{proof}
If \(C\) is empty, return \(R_C=\epsilon\) and \(\lambda_C=0\).
Assume henceforth that \(p<q\).

\emph{The representation and its approximation guarantee.}
Set \(s=W+\lceil2W/\varepsilon\rceil\).  Retain the \(s\)
heaviest positions of \(C\), or all its positions when \(|C|\le s\).
Let \(R_C\) consist of the retained characters in their original
order, and set \(\lambda_C=\mu(C)-\mu(R_C)\).
These choices depend only on \(C\), \(W\), and \(\varepsilon\).

Fix any \(U\) with \(|U|\le W\).
If \(|C|\le s\), the value \(Q(U,C)\) is exact.  Otherwise,
let \(t>0\) be the smallest retained gap mass and denote
\(d=\WED_w(U,C)\).  Every omitted character has mass at most
\(t\), so \cref{obs:target-restoration} gives
\(d\le Q(U,C)\le d+2Wt\).
Every alignment of \(U\) and \(C\) inserts at least \(s-W\) of the
retained target characters, each of mass at least \(t\).  Hence,
\(d\ge(s-W)t\), and
\[
 Q(U,C)\le\left(1+\tfrac{2W}{s-W}\right)d\le(1+\varepsilon)d.
\]

\emph{Constructing and evaluating the representation.}
Precompute prefix masses and a binary segment tree on the positions
of \(Z\).  Each tree node stores a position of maximum gap mass in
its interval, breaking ties by position.  This takes linear time
and space and supports a maximum query in any fragment in
\(O(\log N)\) time.

Compute the reporting count \(\min\{s,|C|\}\) by binary search
in \(O(\log N)\) time.
To report characters of \(C\) in nonincreasing order of mass, put
its interval in a maximum-priority queue, keyed by the mass of its
heaviest character.  Remove the interval with largest key, report
that character, and insert the two nonempty intervals on either side
of its position, obtaining their maxima from the tree.  The queued
intervals partition the unreported positions.  Thus, reporting the
\(s\) heaviest positions, or all positions if fewer exist, takes
\(O((s+1)\log N)\) time.

Sort the retained positions into their original order to obtain
\(R_C\), and compute \(\lambda_C\) from the prefix masses
and the retained characters.  There are at most \(|C|\)
retained positions, so sorting and priority-queue operations contribute
only a \(\log N\) factor, even when \(W/\varepsilon>|Z|\).
The total construction time is \(O((W/\varepsilon)\log N)\).
Given the representation and \(U\), the weighted dynamic-programming algorithm~\cite{WagnerFischer74} takes
\(O(Ws)=O(W^2/\varepsilon)\) time to compute \(Q(U,C)\).
\end{proof}

\subsection{Exact comparisons with one pattern}

When many fragments are compared against the same pattern, we use
multiple-source shortest paths in a planar graph.  This is an
application of Klein's data structure~\cite{Klein05}.

\begin{lemma}[Exact substring queries for one pattern]
\label{lem:one-row-mssp-scan}
Given strings \(P\) and \(Z\), one can preprocess them in
\(O((|P|+1)(|Z|+1)\log N)\) time and space so that, for every
fragment \(Z[p\dd q)\), the distance \(\WED_w(P,Z[p\dd q))\)
can be computed in \(O(\log N)\) time.
\end{lemma}

\begin{proof}
Build the directed edit grid of \(P\) against \(Z\).  It has
\(v\coloneqq (|P|+1) (|Z|+1)\le(N+1)^2\) vertices and edges of non-negative lengths.
The query asks for the shortest-path distance from \((0,p)\) to
\((|P|,q)\), both on the outer face.  Directed monotonicity confines
every such path to the columns of the requested fragment, so this
graph distance is exactly \(\WED_w(P,Z[p\dd q))\).

Add the reverse of every grid edge with cost
\(H:=1+\mu(P)+\mu(Z)\) so that every vertex can then reach every
other vertex.  Each requested distance is less than \(H\), so
none of the underlying shortest paths uses any added edge.

Klein's structure~\cite{Klein05} preprocesses a planar graph with
\(v\) vertices in \(O(v\log N)\) time and space and answers
distances to a vertex on a designated face in logarithmic time.
Apply it to this grid and its outer face.
We use the implementation described by Klein and Mozes
\cite[Sections~7.7--7.9]{KleinMozesDraft}, which handles equal-length
shortest paths without perturbing edge lengths.
\end{proof}

\subsection{Batched gap queries}

The sparse search needs only an oracle for a gap problem: accept
windows sufficiently close to a source string and reject windows beyond a
larger threshold.  Between these thresholds, answers may be arbitrary,
even across repeated queries or different search intervals.
Accepted windows will later receive direct
\(Q\)-certificates from \cref{lem:short-string-approximation}.

For a string \(T\) and a bounded real interval \(I\), a fragment
\(V=T[p\dd q)\) starts in \(I\) when \(\mu_T(p)\in I\).
As usual, \(|I|\) denotes the interval's length.

\begin{lemma}[Batched gap oracle]
\label{lem:batched-radius-test}
Every string \(T\) admits \(\Ohtilde(|T|+1)\)-time and space
preprocessing after which, given an integer \(W\ge1\), a string
\(U\) of length \(|U|\le W\), an accuracy \(0<\zeta\le1/16\),
a radius \(r>0\), and a bounded real interval \(I\), one can
construct an oracle with the following guarantees:
\begin{enumerate}
\item The oracle construction takes
\[
 \Ohtilde\left(
 \min\!\left\{W(|T|+1),\,\tfrac{W^2}{\zeta}\!\left(1+\tfrac{|I|}{r}\right)\right\}
 \right)
\]
time and space.
\item For every fragment \(V=T[p\dd q)\) with \(\mu_T(p)\in I\),
the oracle answers in \(O(\log N)\) time, accepting if
\(\WED_w(U,V)<(1-2\zeta)r\) and rejecting if
\(\WED_w(U,V)\ge r\).
\end{enumerate}
\end{lemma}

The implementation filters light target characters, as in
Kuszmaul's metric edit-distance approximation~\cite[Definition~6.3]{Kuszmaul19}.
The lightness threshold is determined by the requested radius, but a single
data structure supports every threshold.
An interval $I$ containing all target prefix masses gives an oracle for
all fragments of \(T\).  The implementation chooses between preprocessing the
whole filtered target and only the segment needed for \(I\).

\begin{proof}
\emph{Solving the gap problem.}
Put \(t=\zeta r/W\).  For each target fragment \(V=T[p\dd q)\), let
\(L_t(V):=\mu(V)-\mu(V^{>t})\) be the mass removed by filtering,
and define
\[
 \delta_r(U,V):=\WED_w(U,V^{>t})+L_t(V).
\]
The oracle accepts precisely when \(\delta_r(U,V)<r\).
\Cref{obs:target-restoration} gives
\[
 \WED_w(U,V)\le\delta_r(U,V)
 \le\WED_w(U,V)+2\zeta r,
\]
which proves the two gap guarantees.

Put \(K:=W+\lceil W/\zeta\rceil\).
We can reject any fragment containing at least \(K\) retained characters.
Indeed, every alignment of \(U\) and \(V^{>t}\) must then insert at
least \(W/\zeta\) target characters.  Each character has mass greater
than \(t\), so \(\delta_r(U,V)>r\).

\emph{Choosing the filtered target.}
We construct a string \(R\) for exact comparisons with \(U\).
One choice is \(R=T^{>t}\).  For the restricted choice, let
\(a\) and \(b\) be the first and last boundaries of \(T\) whose prefix
masses lie in \(I\).  If there are no such boundaries, no fragment
qualifies, and we return an oracle that always rejects.
Retain every character of mass greater than \(t\) in \(T[a\dd b)\),
followed by the first \(K\) such characters at or after \(b\),
or all of them if fewer exist.  Their concatenation is \(R\).
Since \(\mu(T[a\dd b))\le |I|\), the first part contains
at most \(|I|/t\) characters.  Hence,
\[
 |R|\le |I|/t+K=O(W/\zeta+|I|/t).
\]

\emph{Constructing the oracle.}
For the shared preprocessing, store the prefix masses of \(T\)
and build the maximum tree
described in the proof of \cref{lem:short-string-approximation}.
This takes \(O(|T|+1)\) time and space.
The tree finds the next character of mass greater than \(t\)
at or after any given position in \(O(\log N)\) time:
search from left to right, skipping nodes whose intervals precede
the position or whose maximum mass is at most \(t\), and return
the first qualifying leaf.
Thus, reporting the chosen \(R\) takes
\(O((|R|+1)\log N)\) time.
For the restricted choice, binary searches in the original prefix
masses first locate \(a\) and \(b\).
Binary search also computes the capped cutoff \(\min\{K,|T|+1\}\)
in \(O(\log N)\) time; using this cap leaves both reporting and
rejection unchanged.

Store the original positions of the characters of \(R\) in sorted
order and compute its prefix masses.  Apply
\cref{lem:one-row-mssp-scan} to \(U\) and \(R\).
For the whole filtered target, reporting and preprocessing take
\(\Ohtilde(W(|T|+1))\) time and space.  For the restricted
choice, they take
\[
 \Ohtilde\left(
 W\left(\tfrac{W}{\zeta}+\tfrac{|I|}{t}\right)\right)
 =\Ohtilde\left(\tfrac{W^2}{\zeta}\left(1+\tfrac{|I|}{r}\right)\right).
\]
Choose between the two constructions according to the smaller bound.

\emph{Answering queries.}
For \(V=T[p\dd q)\), binary searches in the stored positions find
boundaries \(x\) and \(y\) such that \(R[x\dd y)\) consists of the
stored characters with original positions in \([p\dd q)\).
If \(y-x\ge K\), reject the query.  Otherwise,
\[
 R[x\dd y)=V^{>t}.
\]
This is immediate when \(R=T^{>t}\).  For the restricted
choice, \(a\le p\le b\).  If \(V\) contained a retained
character omitted from \(R\), it would also contain all \(K\)
stored characters at or after \(b\), contradicting \(y-x<K\).

Query the exact distance from $U$ to \(R[x\dd y)\) and add
\[
 L_t(V)=\mu(T[p\dd q))-\mu(R[x\dd y)).
\]
The resulting value is \(\delta_r(U,V)\), so accepting whenever it
is less than \(r\) gives the required answer.
The binary searches and the exact-distance query each take
\(O(\log N)\) time, and the prefix-mass lookups take constant time.
Both choices use strings of length at most \(N\), so their edit
grids have \(O(N^2)\) vertices and the query bound is uniform over
\(\zeta\).
\end{proof}

\section{Approximating a family of row comparisons}
\label{sec:row-comparisons}
\Cref{sec:comparison-oracles} developed certified primitives for individual and batched
row--window comparisons.  We now use these primitives to construct
certified shortcuts without directly computing every row--window
distance.  The resulting graph approximates the ideal shortcut graph
as stated in \cref{thm:row-comparisons}.


Let \(S,T\) be nonempty strings with \(m=|S|+|T|\le N\).
Fix an integer \(1\le W\le m\)
and any partition
\(0=x_0<\cdots<x_R=|S|\) into rows
\(U_i=S[x_i\dd x_{i+1})\) of at most \(W\) characters, and any
family \(\mathcal C\) of target windows \(T[p\dd q)\), where
\(0\le p\le q\le|T|\).
Write \(M=|\mathcal C|\).  Distinct fragments remain distinct
windows even when they spell the same string.

Let \(G^*\) be the ideal shortcut graph of this partition and window
family: it contains all gap edges in
\([0\dd|S|]\times[0\dd|T|]\) and each shortcut
\((x_i,p)\to(x_{i+1},q)\), for \(V=T[p\dd q)\in\mathcal C\),
at cost \(\WED_w(U_i,V)\).

\begin{restatable}[Approximating row comparisons]{theorem}{rowcomparisons}
\label{thm:row-comparisons}
There is a randomized algorithm that constructs certified shortcuts whose graph
\(G\) includes all gap edges, with the following guarantees:
\begin{enumerate}
\item With high probability, simultaneously for every \(0\le p\le q\le|T|\),
\[
 \operatorname{dist}_G((0,p),(|S|,q))
 \le(3+O(\zeta))\operatorname{dist}_{G^*}((0,p),(|S|,q)).
\]
\item The construction takes, on every outcome, time
\[
 \Ohtilde\left(
 \frac{W^{3/2}\sqrt{mRM}}{\zeta^{5/2}}
 +\frac{mW}{\zeta}+\frac{RW^2}{\zeta^4}+RM\right).
\]
\end{enumerate}
\end{restatable}

We first sample target windows to obtain coarse labels and row radii,
then improve selected labels by a sparse search.

\subsection{Radii and fallback labels}
\label{sec:target-centers}
In this subsection, we assign to each row--window pair a \emph{fallback label}, which is a certified upper bound on its distance obtained using sampled centers.  This label provides a valid shortcut cost even when the pair is not compared directly.  The next lemma bounds its error in terms of a row radius and ensures that, with high probability, only a few windows lie closer than that radius.  The sampling argument is adapted from Andoni~\cite{Andoni20}, while the use of the resulting labels as certified shortcut costs follows~\cite{CDGKS20}.

\begin{lemma}[Sampled radii and fallback labels]
\label{lem:sampled-centers}
Given an integer \(1\le k\le M\), there exists a randomized algorithm that constructs
a radius \(r(U)\) for every row \(U\) and a label \(z(U,V)\) for every
row \(U\) and window \(V\in\mathcal C\), with the following guarantees:
\begin{enumerate}
\item On every outcome, \(0\le r(U)\le\mu(U)\) for every row \(U\),
and, for every row \(U\) and window \(V\in\mathcal C\),
\[
 \WED_w(U,V)\le z(U,V)
 \le\WED_w(U,V)+2(1+\zeta)r(U).
\]
\item With high probability, simultaneously for every row \(U\),
\[
 \bigl|\{V\in\mathcal C:\WED_w(U,V)<r(U)\}\bigr|\le k.
\]
\item The running time on every outcome is
\[
 \Ohtilde\left(
 \frac{mWM}{\zeta k}+RM\right).
\]
\end{enumerate}
\end{lemma}

\begin{proof}
\emph{Constructing the radii.}
Sample
\[
 s=\left\lceil C_{\rm ctr}\frac{M}{k}\right\rceil
          \left\lceil\log_2(N/\zeta)\right\rceil
\]
independent uniform windows from \(\mathcal C\) with replacement.
Here \(C_{\rm ctr}\) is a positive integer constant chosen for the desired
failure probability.  Let \(\mathcal S\) be the multiset of these
samples, retaining repeated draws as separate entries.  Add to
\(\mathcal S\) the empty string as one deterministic virtual center,
with no fixed target location.
Implement the uniform draws with Han and Hoshi's interval
algorithm~\cite[Section~IV]{HanHoshi97}, using unbiased input bits.
The uncapped draw is exactly uniform and needs more than \(b\) bits
with probability at most \(M2^{-b}\). 

To obtain a worst-case running-time bound, stop the sampling procedure
if any draw uses more than
\(C_{\rm uni}\lceil\log_2(N/\zeta)\rceil\) random bits.  Since
\(M\le N^2\) and the number of draws is polynomial in \(N/\zeta\),
choosing \(C_{\rm uni}\) sufficiently large makes the probability that
any draw is stopped at most \((N/\zeta)^{-c}\), for any desired
constant \(c>0\).  The cap bounds the time per draw by a
polylogarithmic factor.  If a draw is stopped, discard all sampled
centers and use only the virtual empty center.  Even in this event,
the construction below returns valid radii and labels.


Preprocess \(T\) for the direct-comparison data structure of
\cref{lem:short-string-approximation}.  For each sampled center
\(C\), construct its short representation \(R_C\) and
nonnegative restoration cost \(\lambda_C\), with accuracy
\(\zeta\).  For the virtual center use
\(R_{\epsilon}=\epsilon\) and \(\lambda_{\epsilon}=0\).
For every row \(i\), compute
\[
 Q(U_i,C)=\WED_w(U_i,R_C)+\lambda_C
 \qquad(C\in\mathcal S).
\]
Choose a minimizing center \(C_i\), breaking ties deterministically,
and set
\[
 r(U_i):=\frac{Q(U_i,C_i)}{1+\zeta}.
\]
The virtual center gives \(0\le r(U_i)\le\mu(U_i)\).
Moreover, for every sampled center \(C\), the approximation guarantee
implies
\[
 r(U_i)\le\frac{Q(U_i,C)}{1+\zeta}\le\WED_w(U_i,C).
\]

\emph{Bounding the number of closer windows.}
Fix a row \(U_i\) and order the \(M\) windows by their true distances
from that row, breaking ties deterministically.
For the uncapped uniform draws, the probability that all samples
miss the first \(k\) windows
is at most
\[
 (1-k/M)^s\le\exp(-sk/M)
 \le(N/\zeta)^{-C_{\rm ctr}}.
\]
On the complementary event, some sampled window \(C\) has rank at most \(k\).
Every window of true distance less than \(r(U_i)\) precedes this
\(C\), since \(r(U_i)\le\WED_w(U_i,C)\).  There are therefore
fewer than \(k\) such windows.  A union bound over the \(R\le N\)
rows proves the simultaneous guarantee for ideal uniform draws.
The bounded implementation agrees with this execution unless a draw
reports failure.  The sampler makes that event arbitrarily unlikely
with an inverse-polynomial bound, so a further union bound proves
the stated guarantee for the implemented algorithm.

\emph{Constructing the labels.}
For every row and window, return
\[
 z(U_i,V):=\WED_w(U_i,R_{C_i})+\WED_w(R_{C_i},V).
\]
The intermediate string \(R_{C_i}\) need not occur in either
input; the shortcut endpoints remain the row and window
boundaries.  The string triangle inequality certifies each label.
Symmetry and the triangle inequality also give
\[
 \begin{split}
 z(U_i,V)
 &\le\WED_w(U_i,V)+2\WED_w(U_i,R_{C_i})\\
 &\le\WED_w(U_i,V)+2Q(U_i,C_i)\\
 &=\WED_w(U_i,V)+2(1+\zeta)r(U_i),
 \end{split}
\]
where the second inequality uses \(\lambda_{C_i}\ge0\).
This is Andoni's center argument~\cite[Section~2.2]{Andoni20},
applied to the short representative.

To obtain the second term of every label, mark the sampled centers
chosen by at least one row.  For each marked \(C\), preprocess
the edit grid of \(R_C\) against \(T\) using
\cref{lem:one-row-mssp-scan}.  For a window \(V=T[p\dd q)\), query
from \((0,p)\) to \((|R_C|,q)\) in this grid.
The answer is exactly \(\WED_w(R_C,V)\).  Query every window
once for this representative and reuse the answers for all assigned
rows.  The first term is already known as
\(Q(U_i,C_i)-\lambda_{C_i}\).  For the virtual center, both
terms follow directly from prefix masses.

\emph{Running time.}
There are \(\Ohtilde(M/k)\) centers,
and their representations have \(O(W/\zeta)\) characters.
For one center, the dynamic programs for all rows take
\[
 O\left(\sum_i |U_i|\,W/\zeta\right)
 =O(|S|W/\zeta).
\]
This uses the actual row lengths; no lower bound on them is needed.
Constructing the representation and its full-target oracle fits
within \(\Ohtilde(mW/\zeta)\) time per center.
At most \(R\) centers are marked, so their window queries cost
\(\Ohtilde(RM)\) in total.  Emitting the labels has
the same bound.  The initial target preprocessing and radius choices
are absorbed in these costs, which proves the running-time claim.
\end{proof}

\subsection{Sparse completion}
\label{sec:sparse-hierarchy}

The fallback labels constructed in \cref{sec:target-centers} provide certified shortcut costs for all row--window pairs, but their additive error is proportional to the corresponding row radius.  When the distance of a pair is much smaller than this radius, its fallback label may therefore be too large.  In this subsection, we identify selected close pairs and assign them certified \((1+\zeta)\)-approximate labels without directly comparing every row with every window.  At each node of a binary tree, the search samples from the rows represented by that node, giving rows of larger radius a greater probability of being selected, and propagates accepted comparisons to its children.  For a fixed path in \(G^*\), the analysis charges the error of the remaining fallback labels to the cost of that path.

The radius-weighted sampling follows Andoni~\cite[Sections~2.1--2.2]{Andoni20} and is the weighted counterpart of the \textsc{SparseStripExtensionSampling} procedure of~\cite{CDGKS20}.  In the weighted setting, target ranges are measured by prefix mass, and the mass length of the range searched for a row is proportional to its radius.  Hence the ratio of the range length to the radius in the running-time bound of \cref{lem:batched-radius-test} is independent of the radius.

Let \(G_0\) be obtained from \(G^*\) by replacing the exact cost of each row--window shortcut by its supplied label \(z(U,V)\).


\begin{restatable}[Completing sparse comparisons]{lemma}{sparsecompletion}
\label{thm:sparse-completion}
Suppose \(0\le r(U)\le\mu(U)\) for every row \(U\), and
the supplied labels satisfy, for every row \(U\) and window \(V\in\mathcal C\),
\[
 \WED_w(U,V)\le z(U,V)
 \le\WED_w(U,V)+2(1+\zeta)r(U).
\]
There exists a randomized algorithm that adds shortcuts to \(G_0\), obtaining
a graph \(G\) with the following guarantees:
\begin{enumerate}
\item Every added shortcut is certified on every outcome.
\item With high probability, simultaneously for every \(0\le p\le q\le|T|\),
\[
 \operatorname{dist}_G((0,p),(|S|,q))
 \le(3+O(\zeta))\operatorname{dist}_{G^*}((0,p),(|S|,q)).
\]
\item If \(\bigl|\{V\in\mathcal C:\WED_w(U,V)<r(U)\}\bigr|\le k\)
for every row \(U\), where \(k\ge1\), the running time on every outcome is
\[
 \Ohtilde\left(\frac{RkW^2}{\zeta^4}+\frac{mW}{\zeta}+RM\right).
\]
\end{enumerate}
\end{restatable}

The sparsity bound \(k\) is used only in the running-time analysis.
For the approximation guarantee, we fix a path in \(G^*\) and show
that, with high probability, the total radius of the rows whose
shortcut on that path does not receive a direct label is at most
\((1+O(\zeta))\) times the cost of the path.


\subsubsection{Sampling rows}
\label{sec:sparse-sampling}

Let \(\mathcal T\) be the binary interval tree on the \(R\)
rows obtained by recursively splitting each nonsingleton interval
into two consecutive intervals whose sizes differ by at most one.
For a node \(v\), let \(\mathcal R(v)\) be its row indices
and let \(t_v:=|\mathcal R(v)|\).  Set
\(L:=1+\lceil\log_2R\rceil\), so the tree has at most
\(L\) levels.  Rows of radius zero
need no improved labels, since their fallback labels are exact.

\begin{lemma}[Sampling by radius]
\label{lem:sparse-row-sampling}
There exists a randomized algorithm that constructs, for each node \(v\), a set
\(\mathcal S_v\subseteq\mathcal R(v)\) with the following guarantees:
\begin{enumerate}
\item On every outcome, each \(\mathcal S_v\) has size
\(\Ohtilde(1/\zeta)\) and contains only rows with positive radius.
\item The total running time is \(\Ohtilde(R/\zeta)\) on every outcome.
\item Fix any set \(\mathcal I\subseteq[0\dd R)\) before sampling.
Let \(\mathcal B\) consist of the rows whose root-to-leaf path
contains a node \(v\) with
\(\mathcal S_v\cap\mathcal I=\varnothing\).
With high probability,
\[
 \sum_{i\in\mathcal B}r(U_i)
 \le(1+O(\zeta))\sum_{i\notin\mathcal I}r(U_i).
\]
\end{enumerate}
\end{lemma}

\begin{proof}
\emph{Exact sampling.}
For each node \(v\), let
\(\sigma_v:=\sum_{i\in\mathcal R(v)}r(U_i)\).
At each node with \(\sigma_v>0\), independently draw
\[
 s:=\left\lceil C\zeta^{-1}\right\rceil
       \left\lceil\log_2(N/\zeta)\right\rceil
\]
rows with replacement, where the positive integer \(C\) is chosen for the desired
failure probability.  Each draw uses the distribution
\[
 \Pr[\text{a draw selects }i]=\frac{r(U_i)}{\sigma_v}
 \qquad(i\in\mathcal R(v)).
\]
Let \(\mathcal S_v\) contain the distinct sampled rows; at a
node with \(\sigma_v=0\), let \(\mathcal S_v=\varnothing\).
For a fixed \(\mathcal I\) and a node \(v\) of positive total
radius, suppose its rows in \(\mathcal I\) have total radius at
least \(\zeta\sigma_v\).  The probability that all draws at
\(v\) miss \(\mathcal I\) is at most
\[
 (1-\zeta)^s\le(N/\zeta)^{-C}.
\]
A union bound over the \(O(R)\) nodes shows that, with high
probability, every such node samples a row from \(\mathcal I\).

\emph{Bounding missed paths.}
On this event, let \(\mathcal F\) be the first nodes with
\(\mathcal S_v\cap\mathcal I=\varnothing\) on root-to-leaf
paths.  Their row sets are disjoint and partition \(\mathcal B\).
For every \(v\in\mathcal F\), the rows outside \(\mathcal I\)
therefore have total radius at least \((1-\zeta)\sigma_v\);
this also holds when \(\sigma_v=0\).  Consequently,
\[
 (1-\zeta)\sum_{i\in\mathcal B}r(U_i)
 =(1-\zeta)\sum_{v\in\mathcal F}\sigma_v
 \le\sum_{i\in\mathcal B\setminus\mathcal I}r(U_i)
 \le\sum_{i\notin\mathcal I}r(U_i).
\]
Dividing by \(1-\zeta\) gives the claimed radius bound for the
exact draws, since \(0<\zeta\le1/16\).

\emph{Bounded implementation.}
Implement each draw with Han and Hoshi's interval
algorithm~\cite[Section~IV]{HanHoshi97}, using unbiased input bits
and the positive radii as weights.  Store cumulative weights without
normalizing them: each input bit halves the current interval inside
\([0,\sigma_v)\), and binary search determines whether this interval
lies entirely within one row's weight interval.
At node \(v\), an exact draw
needs more than \(b\) bits with probability at most \(t_v2^{-b}\).
Cap each draw at \(C_{\rm bit}\lceil\log_2(N/\zeta)\rceil\)
bits.  If any draw exceeds its cap, return empty sets at all nodes.
For a sufficiently large positive integer constant \(C_{\rm bit}\), all
\(O(Rs)\) draws finish within their caps with high probability.
On this event their outputs agree with the exact draws using the
same random bits, so the corresponding sets \(\mathcal B\) also
agree.  A further union bound proves the stated probability guarantee.

Building the tree and its cumulative weights takes
\(O(R\log(R+2))\) time.  Each bounded draw takes
\(O(\log(R+2)\log(N/\zeta))\) time, independently of the
numerical range of the radii.  Removing duplicates fits within
\(\Ohtilde(Rs)=\Ohtilde(R/\zeta)\) time, proving the bound
on every outcome.
\end{proof}

\subsubsection{Propagating anchors}

Fix the sampled-row sets \(\mathcal S_v\) constructed in the previous
subsection.  For a node \(v\), an \emph{anchor} is a triple
\((i,p,q)\), where \(i\in\mathcal S_v\) and \(T[p\dd q)\in\mathcal C\),
representing that the gap oracle accepted the comparison between
\(U_i\) and \(T[p\dd q)\).  Anchors stored at a node determine the
target ranges searched at its children.  At a leaf, every anchor
identifies a row--window pair that receives a direct certified label.
Let \(\mathcal A_v\) denote the set of anchors stored at \(v\).


\begin{lemma}[Propagating sampled comparisons]
\label{lem:sparse-propagation}
Given sets
\(\mathcal S_v\subseteq\{i\in\mathcal R(v):r(U_i)>0\}\)
of size \(\Ohtilde(1/\zeta)\), a deterministic algorithm
constructs anchor sets \(\mathcal A_v\) and adds shortcuts to
\(G_0\), obtaining a graph \(G\), with the following guarantees:
\begin{enumerate}
\item For every node \(v\), row \(i\in\mathcal S_v\), and
      window \(T[p\dd q)\in\mathcal C\) with
      \(\WED_w(U_i,T[p\dd q))<(1-2\zeta)r(U_i)\),
      the anchor \((i,p,q)\) belongs to \(\mathcal A_v\)
      if either \(v\) is the root or \(v\) has a parent \(u\)
      with an anchor \((h,p',q')\in\mathcal A_u\) satisfying
      \begin{equation}
      \label{eq:sparse-propagation-condition}
       \left|\mu_T(p)-(\mu_T(p')+\mu_S(x_i)-\mu_S(x_h))\right|
       \le\frac{Lt_ur(U_i)}{\zeta}.
      \end{equation}
\item Every added shortcut is certified.  At each leaf for
      row \(i\), every anchor \((i,p,q)\) yields the corresponding
      shortcut with a certified \((1+\zeta)\)-approximate label.
\end{enumerate}
If at most \(k\) windows have distance below \(r(U_i)\) from each
row \(i\), where \(k\ge1\), the running time is
\[
 \Ohtilde\left(\frac{RkW^2}{\zeta^4}+\frac{mW}{\zeta}+RM\right).
\]
\end{lemma}

\begin{proof}
\emph{Constructing the anchors.}
All tests for row \(i\) use the gap oracle of
\cref{lem:batched-radius-test} with radius \(r(U_i)\).
Process the tree from its root to its leaves as follows.
\begin{enumerate}
\item At the root, for each sampled row build a gap oracle with
      \(I=[0,\mu(T)]\) and test every window in \(\mathcal C\).
      Store the accepted pairs as anchors.
\item At a child node \(v\) with parent \(u\), for each
      \(i\in\mathcal S_v\) and \((h,p',q')\in\mathcal A_u\),
      form the interval of target starting masses
      \[
       \begin{multlined}
       \biggl[\mu_T(p')+\mu_S(x_i)-\mu_S(x_h)-\frac{Lt_ur(U_i)}{\zeta},\\
       \mu_T(p')+\mu_S(x_i)-\mu_S(x_h)+\frac{Lt_ur(U_i)}{\zeta}\biggr].
       \end{multlined}
      \]
      Its center shifts the anchor's starting mass by the difference
      between the source prefix masses of the two rows.
      For each sampled row, take the union of these intervals and
      merge overlapping intervals into disjoint components.  Build
      a gap oracle for each component and test every window starting
      there.  Store every accepted pair as an anchor.
\item At each leaf for row \(i\), compute the direct
      certificate \(Q(U_i,T[p\dd q))\) with accuracy \(\zeta\)
      for every anchor \((i,p,q)\), using
      \cref{lem:short-string-approximation}.
      Add the corresponding shortcut with this label to \(G_0\).
\end{enumerate}
A comparison below \((1-2\zeta)r(U_i)\) is accepted whenever
queried.  The root searches all windows, and
\cref{eq:sparse-propagation-condition} places the window in a
searched interval at a child.  This proves the first guarantee.
The direct-comparison lemma gives the second.

Sort \(\mathcal C\) once by starting mass, breaking ties by boundary
indices.  Binary search locates the contiguous range of windows
starting in each component, which we enumerate directly.  Since
the components are disjoint, each row--window pair is tested at
most once per node.  A row can be sampled only at nodes on its
root-to-leaf path, so each pair is tested at most \(L\) times in
total.  No consistency between answers in the gap between the two
thresholds is needed.

\emph{Running time.}
Every accepted query has distance below \(r(U_i)\), so under the
degree assumption at most \(k\) windows per row can be accepted.
There are \(\Ohtilde(1/\zeta)\) sampled rows
per node, and hence \(\Ohtilde(k/\zeta)\) anchors per node.
For a fixed sampled child row \(i\), each parent anchor contributes
an interval of length \(2Lt_ur(U_i)/\zeta\).  Their merged union
has \(\Ohtilde(k/\zeta)\) components and total length
\(\Ohtilde(kt_ur(U_i)/\zeta^2)\).
By \cref{lem:batched-radius-test}, constructing the oracles for
these components takes
\[
 \Ohtilde\left(
 \frac{W^2}{\zeta}
 \left(\frac{k}{\zeta}+\frac{kt_u}{\zeta^2}\right)\right)
 =\Ohtilde(kt_uW^2/\zeta^3).
\]
There are \(\Ohtilde(1/\zeta)\) sampled rows in each child.
At each depth, the parent intervals are disjoint, so their spans
\(t_u\), counted once per child, sum to at most \(2R\).
Thus all child preprocessing costs
\(\Ohtilde(RkW^2/\zeta^4)\).
Forming and merging the intervals costs
\(\Ohtilde(Rk/\zeta^2)\) in total.  Each sampled row produces
at most \(k\) anchors per node, so storing all anchors costs
\(\Ohtilde(Rk/\zeta)\).  Both costs fit in the preprocessing bound.

Each root oracle uses the full-target construction, which costs
\(\Ohtilde(mW)\).  The root's \(\Ohtilde(1/\zeta)\)
sampled rows therefore cost \(\Ohtilde(mW/\zeta)\) in total.
Shared preprocessing for the comparison oracles costs
\(\Ohtilde(m)\).  There are at most \(RML\) tests, each taking
logarithmic time, so sorting and testing the windows cost
\(\Ohtilde(RM)\).  Finally, at most \(Rk\) leaf pairs
receive direct certificates, in \(\Ohtilde(RkW^2/\zeta)\) time.
Traversing the tree and reading its sample sets also fit in the
stated bound.
\end{proof}

\subsubsection{Replacing a fixed path in \texorpdfstring{\(G^*\)}{G*}}

Let \(G\) be the graph obtained by applying
\cref{lem:sparse-propagation} to the sample sets from
\cref{lem:sparse-row-sampling}.

\begin{lemma}[Replacing a fixed path in \(G^*\)]
\label{lem:sparse-ideal-path}
Fix \(0\le p\le q\le|T|\) and a path of cost \(d\) in
\(G^*\) from \((0,p)\) to \((|S|,q)\), independently of the
samples.  With high probability, \(G\) contains a path with the
same endpoints and cost at most \((3+O(\zeta))d\).
\end{lemma}

\begin{proof}
\emph{Decomposing the path.}
A path in \(G^*\) uses at most one shortcut per row.  If it
uses a shortcut in row \(i\), write \(V_i=T[p_i\dd q_i)\)
for that shortcut's window.  In all other rows it uses only gap edges.
Slice the path at its endpoints and at each internal source boundary
\(x_i\), immediately before its first edge that consumes a character
of row \(i\).  Expand each shortcut into an optimal alignment,
breaking ties deterministically, and let \(c_i\) be the cost between
consecutive slices.  Slicing before expansion keeps each shortcut
wholly in its row subpath, including initial insertions in its
expansion.  Thus \(c_i\ge0\) and \(\sum_i c_i=d\).

\emph{Choosing the row set.}
For each node \(v\), let
\(C_v:=\sum_{i\in\mathcal R(v)}c_i\).
Let \(\mathcal I\) consist of the rows \(i\) for which
\(c_i<(1-2\zeta)r(U_i)\) and
\[
 C_v\le\frac{Lt_vr(U_i)}{\zeta}
 \qquad\text{for every node \(v\) with \(i\in\mathcal R(v)\)}.
\]
Rows violating the first condition have total radius at most
\(d/(1-2\zeta)\).  Assign each row violating the second condition
to one node \(v\) for which \(r(U_i)<\zeta C_v/(Lt_v)\).
At most \(t_v\) rows are assigned to \(v\), so their radii sum
to at most \(\zeta C_v/L\).  Each row cost contributes to at
most \(L\) nodes, giving \(\sum_v C_v\le Ld\).  Hence
\begin{equation}
\label{eq:sparse-excluded-radius}
 \sum_{i\notin\mathcal I}r(U_i)
 \le\frac{d}{1-2\zeta}+\zeta d
 =(1+O(\zeta))d.
\end{equation}
Every row \(i\in\mathcal I\) has a shortcut in the fixed path:
otherwise \(c_i\ge\mu(U_i)\ge r(U_i)\).  Its window satisfies
\[
 \WED_w(U_i,V_i)\le c_i<(1-2\zeta)r(U_i).
\]

\emph{Propagating the chosen comparisons.}
The set \(\mathcal I\) is fixed before sampling.  By
\cref{lem:sparse-row-sampling,eq:sparse-excluded-radius}, with high
probability the set \(\mathcal B\) from that lemma satisfies
\begin{equation}
\label{eq:sparse-unreached-radius}
 \sum_{i\in\mathcal B}r(U_i)\le(1+O(\zeta))d.
\end{equation}
For every row outside \(\mathcal B\), each node on its
root-to-leaf path samples a row from \(\mathcal I\).

Fix such an outcome and such a tree path.  By
\cref{lem:sparse-propagation}, the root stores the anchor
\((h,p_h,q_h)\) for every sampled row \(h\in\mathcal I\).
Suppose a node \(u\) on the path stores an anchor
\((h,p_h,q_h)\) with \(h\in\mathcal I\), and let
\(i\in\mathcal S_v\cap\mathcal I\) be a sampled row
at its child \(v\) on the path.  The fixed path's subpath between
\((x_h,p_h)\) and \((x_i,p_i)\), taken in their order along the
path, lies in the rows of \(u\)
and has cost at most \(C_u\).  By \cref{lem:potential} and
the definition of \(\mathcal I\),
\[
 \left|\mu_T(p_i)-(\mu_T(p_h)+\mu_S(x_i)-\mu_S(x_h))\right|
 \le C_u\le\frac{Lt_ur(U_i)}{\zeta}.
\]
Thus \cref{lem:sparse-propagation} ensures that \(v\) stores
\((i,p_i,q_i)\).  Induction along the tree path reaches its leaf,
whose only possible sampled row is the leaf's own row.  Therefore
every row outside \(\mathcal B\) receives a direct
\((1+\zeta)\)-approximate label for its window in the fixed path.

\emph{Replacing shortcuts.}
For each row outside \(\mathcal B\), replace its shortcut by
the direct certificate and retain its gap edges, at cost at most
\((1+\zeta)c_i\).  For a row in \(\mathcal B\), use the
supplied fallback shortcut when one exists and retain its gap edges.
The assumed bound on the fallback label gives cost at most
\(c_i+2(1+\zeta)r(U_i)\); rows using only gaps remain unchanged.
These replacements preserve the slicing vertices and hence concatenate.
By \cref{eq:sparse-unreached-radius}, their total cost is at most
\[
 (1+\zeta)d+2(1+\zeta)\sum_{i\in\mathcal B}r(U_i)
 \le(3+O(\zeta))d.\qedhere
\]
\end{proof}

\sparsecompletion*

\begin{proof}
Obtain the sample sets from \cref{lem:sparse-row-sampling},
apply the propagation
algorithm of \cref{lem:sparse-propagation} to these sets, and return its
graph \(G\).  \Cref{lem:sparse-propagation} certifies every added shortcut
and gives the stated running time under the degree assumption;
the \(\Ohtilde(R/\zeta)\) sampling cost is absorbed in this bound.

Before sampling, fix a shortest path in \(G^*\) for every pair of
left and right target endpoints, breaking ties deterministically.
There are at most \((|T|+1)^2\) pairs.  Apply
\cref{lem:sparse-ideal-path} to each path and take a union bound,
using a sufficiently large failure exponent.
This proves the simultaneous distance guarantee.
\end{proof}

\subsection{Completing the row-comparison construction}

\rowcomparisons*

\begin{proof}
If \(M=0\), return an empty shortcut list.  The resulting graph
contains only gap edges and agrees with \(G^*\).
Assume henceforth that \(M\ge1\).

For the running-time analysis, set
\[
 A:=\frac{mW}{\zeta},\qquad B:=\frac{RW^2}{\zeta^4}.
\]
Choose the degree parameter
\[
 k:=\min\left\{M,\left\lceil\sqrt{\frac{AM}{B}}\right\rceil\right\}
   =\min\left\{M,\left\lceil\sqrt{\frac{\zeta^3mM}{RW}}\right\rceil\right\}.
\]
Then \(1\le k\le M\).  Binary search over \([1\dd M]\)
finds the smallest integer satisfying \(k^2RW\ge\zeta^3mM\);
if none exists, use \(k=M\).
Apply \cref{lem:sampled-centers}
to construct radii and fallback labels, then use fresh random bits
for \cref{thm:sparse-completion}.  Return the resulting shortcuts.

By \cref{lem:sampled-centers}(1), the constructed radii and fallback labels satisfy
the assumptions of \cref{thm:sparse-completion} on every outcome. Conditional on these outputs, \cref{thm:sparse-completion} gives the simultaneous
distance guarantee with high probability.
\Cref{lem:sampled-centers}(2) also ensures, with high probability, the degree bound
required for the running-time guarantee of \cref{thm:sparse-completion}. Thus, the total work is
\[
 \Ohtilde\left(\frac{AM}{k}+Bk+RM\right).
\]
Put \(x=\sqrt{AM/B}\).  If \(x\le M\), then
\(k=\lceil x\rceil\) and
\(AM/k+Bk\le2\sqrt{ABM}+B\).
If \(x>M\), then \(k=M\) and
\(AM/k+Bk=A+BM\le A+\sqrt{ABM}\).
Thus the time is
\(\Ohtilde(A+B+\sqrt{ABM}+RM)\), which is the claimed
bound.  The additive terms \(A\) and \(B\) account for the
upper bound \(k\le M\) and for integer rounding,
respectively.

To bound every execution, stop after
\[
 C_{\rm op}\left(
 \frac{AM}{k}+Bk+RM\right)
 \left\lceil\log_2(N/\zeta)\right\rceil^c
\]
RAM operations, where the positive integer constants \(C_{\rm op}\)
and \(c\) are chosen above the bound just proved.
If this limit is reached, return an empty shortcut list.  Otherwise
return the constructed list.  Both outputs are certified.  Whenever the degree bound of
\cref{lem:sampled-centers}(2) holds, the running-time guarantee of \cref{thm:sparse-completion} ensures
that the construction finishes before reaching the operation limit
and therefore returns the constructed shortcut list.  A union bound
with the distance guarantee of \cref{thm:sparse-completion} proves the claimed
high-probability guarantee.
\end{proof}

\section{Constructing local shortcuts}
\label{sec:rectangle-certification}

We combine the row and window construction (\cref{lem:short-pieces-shared-windows}) with the row-comparison algorithm (\cref{thm:row-comparisons}) to obtain certified shortcuts for all target boundary pairs.

\begin{theorem}[Local shortcut construction]
\label{thm:rectangle-certification}
Let \(S\) and \(T\) be nonempty strings with \(m:=|S|+|T|\) and
\(F:=\mu(S)+\mu(T)\), and let \(N\ge m\) be a supplied length bound.
Given $S$, $T$, $N$, and \(0<\varepsilon\le1\), an \(\Ohtilde(mN^{3/5}/\varepsilon^{16/5})\)-time randomized algorithm constructs a
shortcut graph \(G\) for \(S\) and \(T\), with shortcuts stored
explicitly and certified on every outcome, and all gap edges represented implicitly,
such that, with high probability, the following bound holds simultaneously for all
fragments \(T[p\dd q)\):
\begin{equation}
\label{eq:rectangle-certification}
 \operatorname{dist}_G((0,p),(|S|,q))
 \le(3+\varepsilon)\WED_w(S,T[p\dd q))+
       \tfrac{\varepsilon mF}{N}.
\end{equation}
\end{theorem}

The total mass \(F\) affects the additive error, but not the running time.

\begin{proof}
Choose the maximum row length
\[
 W:=\left\lceil\zeta^{3/5} N^{1/5}\right\rceil.
\]
Compute \(W\) by binary search for the least positive integer
satisfying \(W^5\ge\zeta^3N\).  Since \(\zeta\le1\),
we have \(1\le W\le2N^{1/5}\).

\emph{Small inputs.}
If \(m\le\zeta^{-8}\) or \(m<W\), return all
\(|S|\cdot |T|\) unit diagonal
edges of the edit grid, each labelled by its exact substitution or match cost.
Together with the gap edges, they reproduce the full edit
grid and hence give exact distances for every target boundary pair.
This takes \(O(m^2)\) time.  If \(m\le\zeta^{-8}\), then
\(m\le N\) gives
\[
 m^2=m\,m^{3/5}m^{2/5}
 \le mN^{3/5}/\zeta^{16/5}.
\]
If instead \(m<W\), then
\(m^2\le mW=O(mN^{1/5})\), which also fits the claimed bound.
Assume henceforth that \(m>\zeta^{-8}\) and \(W\le m\).

\emph{Construction and approximation.}
Construct the source rows and target windows of
\cref{lem:short-pieces-shared-windows}.  Use
\cref{thm:row-comparisons} to construct certified shortcuts for
these rows and windows, and return them.

On the simultaneous approximation event of that theorem, fix
any target fragment $T[p\dd q)$ and put
\(c=\WED_w(S,T[p\dd q))\).  \cref{lem:short-pieces-shared-windows} supplies,
in the ideal shortcut graph, a path
of cost at most
\[
 (1+O(\zeta))c+O\!\left(\tfrac{\zeta mF}{N}\right).
\]
The guarantee in \cref{thm:row-comparisons} multiplies this cost by at most
\(3+O(\zeta)\).
The distance in the constructed shortcut graph is therefore at most
\[
 (3+O(\zeta))c+O\!\left(\tfrac{\zeta mF}{N}\right).
\]
This bound holds simultaneously for all fragments of \(T\).
Choosing the absolute constant $c_\zeta$ in
\(\zeta=c_\zeta\varepsilon\) sufficiently small gives
\eqref{eq:rectangle-certification}.

\emph{Running time.}
Write \(R\) for the number of rows and \(M\) for the number of
target windows.  \Cref{lem:short-pieces-shared-windows} gives
\[
 R=\Ohtilde\left(\tfrac{m}{W}\right)\qquad\text{and}\qquad
 M=\Ohtilde\!\left(\tfrac{N}{\zeta^2W}\right).
\]
Substituting these bounds in
\cref{thm:row-comparisons} gives total time, on every outcome,
\begin{equation}
\label{eq:rectangle-cost}
 \Ohtilde\left(
 m\left(\frac{\sqrt{NW}}{\zeta^{7/2}}
       +\frac{W}{\zeta}+\frac{W}{\zeta^4}
       +\frac{N}{\zeta^2W^2}\right)\right).
\end{equation}
The window construction costs \(\Ohtilde(m/\zeta)\), which is
absorbed by the second term.

Our choice of \(W\) balances the first and last terms of
\eqref{eq:rectangle-cost}.
Since \(N\ge m>\zeta^{-8}\), the quantity inside the ceiling
exceeds one, so rounding changes these bounds by only absolute
factors.  Substitution bounds the first and last terms by
\(\Ohtilde(mN^{3/5}/\zeta^{16/5})\).
The third term of \eqref{eq:rectangle-cost} is
\(\Ohtilde(mN^{1/5}/\zeta^{17/5})\); its ratio to the target
bound is \(O(1/(\zeta^{1/5}N^{2/5}))=O(1)\), again using
\(N>\zeta^{-8}\).  The second term is at most the third, so both middle terms are absorbed by the target bound.
Finally,
\(\zeta=c_\zeta\varepsilon\), with an absolute constant
\(c_\zeta\), converts these bounds to the stated dependence
on~\(\varepsilon\).
\end{proof}

\section{The algorithm for a distance guess}
\label{sec:approximation-algorithm}

The global algorithm for a supplied budget \(D>0\) has three steps: cover
inexpensive alignments (\cref{lem:alignment-rectangle-cover}), construct shortcuts inside each rectangle (\cref{thm:rectangle-certification}),
and combine the shortcuts (\cref{prop:implicit-weighted-shortcut-consumer}).
Choosing the budget is the subject of \cref{sec:aspect-reduction}.

\begin{theorem}[Approximation from a distance guess]
\label{thm:weighted-cdgks-upper-bound-primitive}
Given strings \(X\) and \(Y\), a length bound \(N\ge|X|+|Y|\),
a metric cost oracle \(w\), a budget \(D>0\), and
\(0<\varepsilon\le1\), a randomized \(\Ohtilde(N^{8/5}/\varepsilon^{16/5})\)-time algorithm returns a value
\(Q(D)\) with the following guarantees:
\begin{enumerate}
\item On every outcome, \(Q(D)\ge\WED_w(X,Y)\).
\item If \(D\ge\WED_w(X,Y)\), then
\(Q(D)\le(3+\varepsilon)\WED_w(X,Y)+\varepsilon D\) with high probability.
\end{enumerate}
\end{theorem}

\begin{proof}
\emph{The algorithm.}
Handle empty strings exactly.  Otherwise, set \(\eta=c_\eta\varepsilon\),
where \(0<c_\eta\le1\) is a sufficiently small absolute rational constant chosen
below, and proceed as follows.
\begin{enumerate}
\item Apply \cref{lem:alignment-rectangle-cover} with budget \(D\).
\item For every unit rectangle, include its exact diagonal edge.
      For every other rectangle \([a\dd b]\times[\ell\dd r]\)
      with nonempty projections, apply \cref{thm:rectangle-certification}
      to \(S=X[a\dd b)\) and \(T=Y[\ell\dd r)\), with length
      bound \(N\) and accuracy \(\eta\).  Translate each returned
      shortcut endpoint \((x,p)\) to \((a+x,\ell+p)\), keeping its label.
      Rectangles with an empty projection require only gap edges.
\item Apply \cref{prop:implicit-weighted-shortcut-consumer} to all
      returned shortcuts and exact diagonal edges, and return its
      value as \(Q(D)\).
\end{enumerate}

\emph{Running time and certification.}
For rectangle \(s\), let \(m_s\) be its total character count.
The overlap bounds of \cref{lem:alignment-rectangle-cover} give \(\sum_s m_s=O(N)\).  Hence, the total
local construction time is
\[
 \sum_s\Ohtilde(m_sN^{3/5}/\eta^{16/5})
 =\Ohtilde(N^{8/5}/\varepsilon^{16/5}).
\]
The number of shortcut records is bounded by the same quantity,
and translating their endpoints takes constant time per record.
The dynamic programming algorithm of \cref{prop:implicit-weighted-shortcut-consumer} adds a logarithmic overhead hidden in the \(\Ohtilde\) notation.  Every
shortcut is certified on every outcome, so the returned value $Q(D)$ is
an upper bound on the true distance $\WED_w(X,Y)$.

\emph{Approximation.}
A union bound over the \(O(N)\) local calls gives a high-probability event on which all local guarantees hold; we henceforth condition on this event.
Let \(d:=\WED_w(X,Y)\) and suppose that \(d\le D\).
Let \(\Gamma\) be an optimal alignment of cost \(d\).
The cover contains every edge of \(\Gamma\), and each edge that
consumes a source character belongs to its unique source rectangle.

For a rectangle with nonempty source projection \(S=X[a_s\dd b_s)\), take the
subpath starting immediately before the edge that consumes
\(X[a_s]\) and ending immediately after the edge that consumes
\(X[b_s-1]\).  Both endpoints belong to the rectangle by coverage,
and monotonicity places the entire subpath inside it.  These
subpaths are edge-disjoint and occur in source order.  Intervening
edges are insertions, which remain available in the shortcut graph as gap edges.

For each rectangle \(s\) handled by \cref{thm:rectangle-certification}, let \(c_s\) be
this subpath's cost and \(F_s\) the total mass of the rectangle's
projections.  The edit distance between the corresponding substrings
is at most \(c_s\), so \cref{thm:rectangle-certification} gives a replacement with
the same endpoints and cost at most
\[
 (3+\eta)c_s+\tfrac{\eta m_sF_s}{N}.
\]
All remaining parts of \(\Gamma\) use
gap edges or exact unit-rectangle diagonals; we denote their total cost by $c_0$.
The replacements
concatenate with these edges.  Since \(c_0+\sum_s c_s=d\), the total cost of the resulting path is at most
\[
 (3+\eta)d+\tfrac{\eta}{N}\sum_s m_sF_s.
\]
For every such rectangle, \(F_s=O(D)\).  Combining this with
\(\sum_s m_s=O(N)\) bounds the last term by \(O(\eta D)\).
Since \(\eta=c_\eta\varepsilon\), choosing \(c_\eta\) sufficiently small
proves the claimed guarantee.
\end{proof}

\section{Choosing the distance budget}
\label{sec:aspect-reduction}

In this section, we use the following lemma to remove the distance guess required by
\cref{thm:weighted-cdgks-upper-bound-primitive}.

\begin{restatable}[Computing a coarse distance interval]{lemma}{coarsedistanceinterval}
\label{lem:aspect-free-coarse-bracket}
Given distinct nonempty strings \(X\) and \(Y\) of total length \(N\), a randomized
\(O(N\log^2 N)\)-time algorithm returns an interval \([L,U]\subseteq (0,\infty)\)
of \emph{multiplicative width} \(U/L=O(N)\) such that, with high probability,
\(L\le\WED_w(X,Y)\le U\).
\end{restatable}

Trying \(O(\log N)\) successively doubled budgets starting at \(L\) then proves the
main theorem.  Our proof of \cref{lem:aspect-free-coarse-bracket} adapts Kuszmaul's
\cite[Section~6]{Kuszmaul19} \(\Ohtilde(N^{2-\delta})\)-time
\(O(N^\delta)\)-approximation algorithm at the near-linear-time endpoint \(\delta=1\).
We repeat the necessary details to eliminate the assumption \(|X|=|Y|\)
and to correct a minor gap in the original arguments.

\subsection{Random simplification and capped distances}

We use the filter \(Z^{>r}\) defined in \cref{sec:model}:
it deletes every character of gap mass at most \(r\).
Kuszmaul's approximation relies on the following lemma,
whose original statement \cite[Lemma~6.5]{Kuszmaul19} includes an unnecessary assumption $|X|=|Y|$.
We reproduce the proof below for completeness.

\begin{lemma}[Random simplification]
\label{lem:kuszmaul-random-simplification}
Consider strings $X,Y\in \Sigma^*$ and a threshold \(R>0\).
If \(r\) is chosen uniformly from \([R,2R)\), then
\[
 \mathbb E_r\bigl[
   \WED_w(X^{>r},Y^{>r})
 \bigr]
 \le 5\WED_w(X,Y).
\]
\end{lemma}

\begin{proof}
Filter an optimal alignment as follows.  A diagonal becomes a
gap edge if exactly one of its characters survives; edges with no
surviving character are omitted.  This yields an alignment of
\(X^{>r}\) and \(Y^{>r}\).  Original gap edges do not increase
in cost.

Consider a diagonal pairing \(a\) with \(b\), and assume by symmetry that
\(g(a)\le g(b)\).  Its projected cost \(C_r\) is \(w(a,b)\)
if both characters survive, \(g(b)\) if only \(b\) survives,
and zero otherwise.  The reverse triangle inequality gives
\[
 \Pr[g(a)\le r<g(b)]
 \le\tfrac{g(b)-g(a)}{R}
 \le\tfrac{w(a,b)}{R}.
\]
If \(g(b)\le4R\), then
\[
 \mathbb E_r[C_r]
 \le w(a,b)+g(b)\Pr[g(a)\le r<g(b)]
 \le5w(a,b).
\]
If \(g(b)>4R\), then whenever only \(b\) survives,
\(g(a)\le r<2R<\frac12g(b)\), so
\(w(a,b)\ge g(b)-g(a)>\frac12g(b)\).
Thus, \(C_r\le2w(a,b)\) on every outcome in this case.

The filtered distance is at most the cost of the projected alignment.
Summing the expected costs of its edges gives
\[
 \mathbb E_r[\WED_w(X^{>r},Y^{>r})]\le5\WED_w(X,Y).\qedhere
\]
\end{proof}

An alignment cheaper than every gap cost can contain only
substitutions and matches.  This allows us to compute distances below a given
threshold exactly.  The next observation is the no-gap case of Kuszmaul's
low-distance observation~\cite[Observation~6.2]{Kuszmaul19}.

\begin{observation}[Distances below the minimum gap cost]
\label{lem:capped-low-distance}
Consider strings \(U,V\in \Sigma^*\) and a threshold \(R>0\).  If every character of
\(U\) and \(V\) has gap mass at least \(R\), then
\[
 \min\{\WED_w(U,V),R\}
\]
can be computed in \(O(|U|+|V|+1)\) time.
\end{observation}

\begin{proof}
If the strings have different lengths, every alignment contains a
gap operation of cost at least~\(R\), so return \(R\).
Otherwise, compute
\[
 e=\sum_{i=0}^{|U|-1}w(U[i],V[i])
\]
and return \(\min\{e,R\}\).  An alignment of cost below \(R\)
contains no gap operations, so its cost must be \(e\).
Thus, the returned value is \(\WED_w(U,V)\) when that distance is
below \(R\), and is \(R\) otherwise.  The computation is a single
scan.
\end{proof}

\subsection{Testing a distance scale}
By Markov's inequality, a sufficiently small distance remains below the
chosen scale after random simplification with constant probability.
\Cref{lem:capped-low-distance} lets us recognize this event in linear
time.  Repetition yields a scale test: a \textnormal{\textsc{Close}}
answer always certifies an upper bound,
while a \textnormal{\textsc{Far}} answer gives a lower bound with high probability.
This test adapts the gap test in Kuszmaul's proof of
Theorem~6.1~\cite[Appendix~D]{Kuszmaul19} to the near-linear-time endpoint.

\begin{lemma}[Testing a distance scale]
\label{lem:coarse-scale-test}
Given strings \(X,Y\) of total length \(N\) and a scale \(R>0\),
a randomized \(O(N\log N)\)-time algorithm returns either \textnormal{\textsc{Close}} or \textnormal{\textsc{Far}} with the
following guarantees:
\begin{enumerate}
\item On every outcome, if the algorithm returns \textnormal{\textsc{Close}}, then \(\WED_w(X,Y)<3NR\).
\item If \(\WED_w(X,Y)<R/10\), then the algorithm returns \textnormal{\textsc{Close}} with high probability.
\end{enumerate}
\end{lemma}

\begin{proof}
\emph{The algorithm.}
Make
\(t=c\lceil\log_2N\rceil\) independent threshold draws, where
\(c\) is a sufficiently large integer constant chosen below.  For each drawn threshold
\(r\in[R,2R)\), form \(X^{>r},Y^{>r}\) and let \(d_r\) denote
their distance.  Use \cref{lem:capped-low-distance} to compute
\(h=\min\{d_r,R\}\).
It is applicable because every surviving character has gap mass
greater than \(r\ge R\).  The test returns \textnormal{\textsc{Close}} if some \(h<R\),
and returns \textnormal{\textsc{Far}} otherwise.

To implement a draw, partition \([R,2R)\) at the gap masses of input
characters.  This gives at most \(N+1\) half-open intervals on each
of which both simplified strings are fixed.  Use Han and Hoshi's
interval algorithm~\cite[Section~IV]{HanHoshi97} to sample an interval
with probability equal to its length divided by \(R\), and take
its left endpoint as the threshold.  The uncapped procedure produces
exactly the same distribution of simplified strings as a uniform
real threshold.  Cap each draw at \(q=\lceil\log_2N\rceil+2\)
random bits and use \(r=R\) on truncation.  The probability of
truncation is at most \(N2^{-q}\le\tfrac14\).

\emph{Correctness.}
Denote \(d=\WED_w(X,Y)\).
Every \textnormal{\textsc{Close}} answer certifies \(d<3NR\).  Indeed, a witness
\(h<R\) equals \(d_r\).  Restoring the deleted characters costs less
than \(2NR\), so the triangle inequality gives
\[
 d<h+2NR<R+2NR\le3NR.
\]

Now, suppose \(d<R/10\).  For an uncapped draw,
\cref{lem:kuszmaul-random-simplification} gives
\[
 \mathbb E[d_r]<\tfrac{R}{2}.
\]
Markov's inequality gives \(d_r<R\), and hence a witness for \textnormal{\textsc{Close}},
with probability greater than \(\tfrac12\).  The bounded draw produces
the same witness unless it is truncated, an event of probability
at most \(\frac14\).
Thus, each bounded draw produces a witness for \textnormal{\textsc{Close}} with probability
greater than \(\frac14\), and all \(t\) independent samples fail to produce
such a witness with probability at most \((3/4)^t\).  Choosing
\(c\) sufficiently large makes this failure probability an arbitrarily
small inverse polynomial in \(N\).

\emph{Running time.}
Sorting the gap masses and forming the sampling intervals costs
\(O(N\log N)\) time.  Binary search among the interval boundaries
implements each bounded draw in \(O(\log^2 N)\) time.  Forming the
simplified strings and applying \cref{lem:capped-low-distance} takes
\(O(N)\) time per draw.
The \(t=O(\log N)\) draws therefore take
\(O((N+\log^2 N)\log N)=O(N\log N)\) time in total.
\end{proof}

\subsection{Bounding the distance using character gap masses}
The algorithm behind \cref{lem:aspect-free-coarse-bracket} binary-searches
the input gap masses and their halves using the scale test of
\cref{lem:coarse-scale-test}.  Consecutive candidate scales can differ
by an arbitrarily large factor, so the resulting distance bounds need
not have multiplicative width \(O(N)\).  The following lemma, which
we believe to be missing in~\cite{Kuszmaul19}, handles this difficulty
by exploiting an interval of thresholds throughout which the
simplified strings remain unchanged.

\begin{lemma}[Distance bounds from a fixed simplification]
\label{lem:coarse-gap-bracket}
Given strings \(X,Y\) of total length \(N\) and thresholds \(0<A<B\) such
that no input character has gap mass in \((A,2B)\), a deterministic
\(O(N)\)-time algorithm returns an interval \([L,U]\) with
\(0<L\le U\) satisfying the following guarantees:
\begin{enumerate}
\item The multiplicative width \(U/L\) is at most \(15N\).
\item If \(A/10\le\WED_w(X,Y)<3NB\), then
      \(L\le\WED_w(X,Y)\le U\).
\end{enumerate}
\end{lemma}

\begin{proof}
Form \(X^{>A}\) and \(Y^{>A}\).  Every surviving character has
gap mass at least \(2B\), so we can use \cref{lem:capped-low-distance}
to compute
\[
 e=\min\{\WED_w(X^{>A},Y^{>A}),B\}.
\]
Return \([L,U]\), where
\[
 L=\max\{A/10,e/5\}\qquad\text{and}\qquad U=15NL.
\]
The computation takes \(O(N)\) time, and the interval satisfies
\(0<L\le U\) and \(U/L=15N\).

For containment, suppose that
\(A/10\le d:= \WED_w(X,Y)<3NB\).  Every threshold \(r\in[A,2A)\) produces
the same simplified strings, so \cref{lem:kuszmaul-random-simplification}
gives \(e\le5d\).  Hence, \(L\le d\).
If \(e=B\), then \(d<3NB=3Ne\le15NL=U\).
Otherwise, \(e=\WED_w(X^{>A},Y^{>A})\), and restoring
the deleted characters costs at most \(NA\).  Thus,
\[
 d\le e+NA\le5L+10NL\le15NL=U.\qedhere
\]
\end{proof}

We now combine \cref{lem:coarse-gap-bracket} with the scale test of \cref{lem:coarse-scale-test}.
Binary search leaves two consecutive candidate scales.  If they differ
by at most a factor of two, the scale-test bounds suffice; otherwise,
the preceding lemma supplies an interval of multiplicative width \(O(N)\).

\coarsedistanceinterval*

\begin{proof}
\emph{The algorithm.}
Form and sort the following \(O(N)\) distinct candidate scales
\[
 \mathcal S:=\{g(a)/2,g(a):a\text{ occurs in }X\text{ or }Y\},
\]
and initialize \(A=\min\mathcal S\) and \(B=\max\mathcal S\).
Thus, \(2A\) and \(B\) are the smallest and largest input gap masses.
Use \cref{lem:capped-low-distance} to compute
\(e=\min\{\WED_w(X,Y),2A\}\).
If \(e<2A\), return \([e,e]\).

While \(A\) and \(B\) are not consecutive in \(\mathcal S\), choose a median
\(R\in \mathcal S\) of the candidate scales strictly between them and invoke
\cref{lem:coarse-scale-test} with fresh randomness.
Set \(A\coloneqq R\) if the answer is \textnormal{\textsc{Far}}, and set
\(B\coloneqq R\) if it is \textnormal{\textsc{Close}}.
When the search ends, return \([A/10,3NB]\) if \(B\le2A\).
Otherwise, return the interval supplied by
\cref{lem:coarse-gap-bracket} for these endpoints.

\emph{Correctness.}
Let \(d=\WED_w(X,Y)>0\).  If the initial check returns, its
interval is \([d,d]\).  Otherwise, the initial endpoints satisfy
\(2A\le d\le NB\).
For the analysis, fix independent random choices for the test at
each candidate scale, revealing them only when that scale is queried.
Every \textnormal{\textsc{Close}} answer deterministically satisfies \(d<3NR\).
By \cref{lem:coarse-scale-test} and a union bound over the
\(O(N)\) candidate scales, with high probability, every
\textnormal{\textsc{Far}} answer satisfies \(R/10\le d\).
On this event, the bounds
\[
 A/10\le d<3NB
\]
hold initially and are preserved by every update.

If the final endpoints satisfy \(B\le2A\), the returned interval
has multiplicative width at most \(60N\) and contains \(d\)
whenever these bounds hold.
Otherwise, there is no input gap mass in \((A,2B)\): a mass in
\((A,B)\) would itself be a candidate scale strictly between the
endpoints, and a mass in \([B,2B)\) would have its half strictly
between them.  Both contradict the termination condition.
Thus, \cref{lem:coarse-gap-bracket} applies and returns an interval
of multiplicative width at most \(15N\) that contains \(d\)
whenever the endpoint bounds hold.  Hence, the width is at most
\(60N\) on every outcome, and containment holds with high probability.

\emph{Running time.}
Sorting and the capped-distance computation take \(O(N\log N)\) time.
Each query halves the number of candidate scales strictly between
\(A\) and \(B\), so the search makes \(O(\log N)\) calls,
each taking \(O(N\log N)\) time.  The final application of
\cref{lem:coarse-gap-bracket}, if needed, takes \(O(N)\) time.
The total is \(O(N\log^2 N)\).
\end{proof}

\subsection{Dyadic budget selection}

We complete the proof of the main result by doubling budgets from the lower endpoint of the interval produced by \cref{lem:aspect-free-coarse-bracket}.
On the event that the interval contains the distance, one tested budget
lies between the distance and twice the distance.  This converts the
supplied-budget guarantee into the desired relative approximation.

\mainresult*

\begin{proof}
Handle empty or equal strings exactly.  Otherwise, apply
\cref{lem:aspect-free-coarse-bracket} to obtain an interval \([L,U]\) of multiplicative width $U/L = O(N)$, and construct a set of budget guesses
\[
 \mathcal D=\{L,2L,\ldots,2^JL\},
 \qquad\text{where}\qquad
 J=\min\{j\in\mathbb Z_{\ge0}:2^jL\ge U\}.
\]
There are \(O(1+\log(U/L))=O(\log N)\) guesses; repeated doubling constructs them in $O(\log N)$~time.

For each guess \(D\in\mathcal D\), invoke
\cref{thm:weighted-cdgks-upper-bound-primitive} on the original
\(X,Y,w\), with length bound \(N\), budget \(D\), fresh randomness,
and accuracy \(\xi=\varepsilon/3\).
Return \(Q=\min_{D\in\mathcal D}Q(D)\).  Every obtained answer is
a certified upper bound for the original instance, regardless of
its budget and random choices.  Thus, \(Q\) is a certified
upper bound even if the interval \([L,U]\) does not contain the distance.

On the event that \([L,U]\) contains the distance, put
\(d=\WED_w(X,Y)>0\) and let \(D_\star\) be the smallest member
of \(\mathcal D\) at least \(d\).  Then,
\[
 d\le D_\star<2d.
\]
Conditional on the computed interval, this budget is fixed before \cref{thm:weighted-cdgks-upper-bound-primitive} uses its fresh randomness.  Its approximation guarantee
therefore gives, with high probability,
\[
 Q\le Q(D_\star)
 \le3d+\xi(d+D_\star)
 \le(3+3\xi)d=(3+\varepsilon)d.
\]
Choose the failure exponents so that both guarantees hold with the
prescribed high probability.

The algorithm of \cref{lem:aspect-free-coarse-bracket} takes \(O(N\log^2 N)\) time.
Because \(\xi\) is an absolute constant times~\(\varepsilon\),
each application of \cref{thm:weighted-cdgks-upper-bound-primitive} takes
\[
 \Ohtilde(N^{8/5}/\varepsilon^{16/5})
\]
time.
Summing over the \(O(\log N)\) guesses proves the claimed total running time.
\end{proof}

\appendix
\section{Finite-precision implementation}
\label{sec:finite-precision}

In this appendix, we justify the word-RAM implementation stated in the introduction.
The machine supports multiplication and integer division on words.
Rational quantities are represented exactly by integer numerators
and denominators. No costs or distances are rounded.

\begin{lemma}[Bounded bit length]
\label{lem:bounded-precision}
Let \(0<\varepsilon\le1\) be an inverse power of two, and let
\(B\ge1\) be an integer.
If metric queries return integers in \([0\dd 2^B)\), the algorithm of
\cref{cor:weighted-cdgks-scale-selection} can be implemented using
integers of \(B+O(\log(N/\varepsilon))\) bits, with a constant number
of integer operations per real arithmetic operation.
\end{lemma}

\begin{proof}
We first bound the integral distance computations, then the rational
quantities used in thresholding and sampling.

\emph{Integral quantities.}
All strings used by the algorithm are subsequences of the input
strings.  Prefix masses and exact edit distances are therefore
integers of magnitude \(O(N2^B)\).  The same bound applies to the
restoration costs and direct certificates of
\cref{lem:short-string-approximation}, and to the fallback labels of
\cref{lem:sampled-centers}.  The sums and differences used by the
shortcut algorithm of \cref{prop:implicit-weighted-shortcut-consumer}
have magnitude at most \(2^BN^{O(1)}\).
All these quantities require \(B+O(\log N)\) bits.

\emph{Planar distance queries.}
We must also bound the temporary values used inside the planar
oracle of \cref{lem:one-row-mssp-scan}.  Its augmented edit grids
have \(N^{O(1)}\) vertices and integer edge lengths of magnitude
\(O(N2^B)\).  Klein's preprocessing~\cite{Klein05} changes the
root of a maintained shortest-path tree.  During each root change,
one edge length varies between a negative shortest-path distance
and its original length.  Thus, the current tree-path lengths and
\emph{slacks}, which are edge lengths plus differences of
root-to-vertex path lengths, have magnitude \(2^BN^{O(1)}\).
Their updates use addition, subtraction, and comparison
\cite[Sections~7.7 and~7.9]{KleinMozesDraft}.
The accumulated updates stored in the dynamic trees have the same
magnitude bound, since there are only polynomially many updates.
Ties require no numerical perturbation: among edges of minimum
slack, we can choose the one farthest from the outer face along the
maintained dual-tree path, whose vertices represent faces
\cite[Section~7.8]{KleinMozesDraft}.

\emph{Thresholds and radii.}
The internal accuracy parameters are fixed rational multiples of
\(\varepsilon\), so their numerators and denominators have
\(O(\log(2/\varepsilon))\) bits.
We show that the remaining rational quantities have numerators of
\(B+O(\log(N/\varepsilon))\) bits and denominators of
\(O(\log(N/\varepsilon))\) bits.
The coarse interval and dyadic budgets in \cref{sec:aspect-reduction}
use integer mass and distance sums, division by fixed integers,
and \(O(\log N)\) doublings, so they satisfy these bounds.

For the window construction of \cref{sec:short-pieces-shared-windows},
recall that \(F\) is the total gap mass of the local strings,
\(W\le N\) is the maximum row length, \(\zeta=\Theta(\varepsilon)\)
is the local accuracy parameter, and
\(\Lambda=O(\log(N/\varepsilon))\) is the positive integer normalization parameter.
The mass resolution is \(\Delta=\zeta FW/(N\Lambda)\).
Write \(\zeta=a/b\) with positive integers \(a\) and \(b\) of
\(O(\log(2/\varepsilon))\) bits, and put \(K=\lceil1/\zeta\rceil\).
At dyadic level \(i\), the offsets in the proof of
\cref{lem:shared-target-windows} have the form
\[
 2^i\Delta\left(1+\frac jK\right)
 =\frac{aFW\,2^i(K+j)}{bN\Lambda K},
 \qquad j\in[0\dd K).
\]
Here, \(K=O(1/\varepsilon)\) and
\(2^i=O(N\Lambda/(\zeta W))\).
Since \(F=O(N2^B)\), the displayed numerator and denominator
satisfy the stated bit bounds.  These offsets select window
boundaries; they never become edge lengths.

The radius \(r(U)\) assigned to a row \(U\) by
\cref{lem:sampled-centers} is an integer direct certificate divided
by \(1+\zeta\).  The filtering thresholds and search endpoints in
\cref{lem:batched-radius-test,lem:sparse-propagation}
use these radii scaled by \(\zeta\), \(1/\zeta\),
polynomially bounded integer factors, or \(1/W\).
The search intervals are centered at sums and differences of original
integer prefix masses.  Each level therefore forms its endpoints
anew, rather than scaling endpoints inherited from its parent.
Thus, denominators do not accumulate across levels, and these
quantities satisfy the same bit bounds.

For the target \(T\), radius \(r>0\), and starting interval \(I\)
of \cref{lem:batched-radius-test}, we choose between the two
preprocessing bounds by comparing \(\zeta r(|T|+1)\) with
\(W(r+|I|)\).  This avoids forming \(|I|/r\); both expressions
satisfy the same numerator and denominator bounds.

\emph{Sampling and integer parameters.}
In the radius-weighted sampling of \cref{lem:sparse-row-sampling},
we cancel the common factor \(1/(1+\zeta)\) and use integer direct
certificates as weights.  The interval sampler compares cumulative
weights with dyadic subintervals of the total weight.
It uses at most \(O(\log(N/\varepsilon))\) random bits per draw.
Clearing the resulting powers of two therefore gives comparisons
between \(B+O(\log(N/\varepsilon))\)-bit integers.
The threshold sampler in the proof of \cref{lem:coarse-scale-test}
is called at integer or half-integer scales by
\cref{lem:aspect-free-coarse-bracket}.  Its interval lengths are
therefore multiples of \(1/2\); multiplying them by two gives
integer weights.  The uniform center sampling of
\cref{lem:sampled-centers} uses unit weights.
None of these implementations stores normalized probabilities.

Rounded logarithms are computed by doubling; ceilings and integer
roots are obtained by integer searches, as discussed in \cref{sec:model}.
All counts, indices, and operation limits are polynomially bounded
in \(N/\varepsilon\) on every random outcome.
Fractions need not be reduced: comparisons use cross multiplication,
and scaling uses the parameter factors described above.
Each cross product contains only one numerator and one denominator,
so it still has \(B+O(\log(N/\varepsilon))\) bits.
Thus, each operation uses a constant number of integer operations
within the claimed bit bound.
\end{proof}

\begin{corollary}[Word-RAM implementation]
\label{cor:word-ram}
Let \(0<\varepsilon\le1\) be a rational occupying \(O(1)\) words
on a word RAM with \(O(\log(N/\varepsilon))\)-bit words.
Suppose each metric query takes constant time and returns a
nonnegative integer that fits in one machine word.
There is a randomized algorithm on this machine that has the
approximation and probability guarantees of
\cref{cor:weighted-cdgks-scale-selection} and runs in
\(\Ohtilde(N^{8/5}/\varepsilon^{16/5})\) time on every random outcome,
with hidden constants independent of \(\varepsilon\).
\end{corollary}

\begin{proof}
Let \(\ell\) be the word length.  The machine can address the input,
and \(\varepsilon\) has an \(O(1)\)-word rational representation,
so \(\log(N/\varepsilon)=O(\ell)\).
By repeated halving, choose an inverse power of two
\(\widehat\varepsilon\in[\varepsilon/2,\varepsilon]\).
Apply \cref{lem:bounded-precision} with accuracy
\(\widehat\varepsilon\) and \(B=\ell\).
Every integer in the resulting implementation has
\[
 \ell+O(\log(N/\widehat\varepsilon))=O(\ell)
\]
bits and hence occupies \(O(1)\) words.  Each arithmetic operation
therefore takes constant time in the stated model.
Since \(\widehat\varepsilon\le\varepsilon\), the approximation
factor is at most \(3+\varepsilon\).  Since
\(\widehat\varepsilon\ge\varepsilon/2\), the accuracy change affects
the running time by only a constant factor.
The time and probability guarantees now follow from
\cref{cor:weighted-cdgks-scale-selection}.
\end{proof}

\section*{Disclosure of AI assistance}

This work developed through extensive interaction with OpenAI's GPT-5.6
Sol and GPT-6 Astra.  Running-time exponents refer to input length for
fixed approximation parameters.

We asked Sol to develop a truly subquadratic constant-factor approximation
for metric weighted edit distance.  Following author guidance,
an overnight autonomous run produced an \(N^{40/21+o(1)}\)-time algorithm
that was, in retrospect, overcomplicated.  We then asked Sol to explain its
components and pointed out many inefficiencies.  Our most important
suggestion was to shorten the target $V$ of \(\WED_w(U,V)\) computations to \(O(|U|)\)
characters, with an additive correction for discarded mass
(\cref{lem:short-string-approximation}).  This helped yield an
\(\Ohtilde(N^{16/9})\)-time algorithm closely resembling~\cite{CDGKS20}.

The large number of target windows, and hence of sampled centers (called
pivots in~\cite{CDGKS20}), was the main bottleneck.  To address it, we
proposed organizing rows using few mass levels
(\cref{lem:local-source-partition}) and processing
independent rectangles (\cref{lem:alignment-rectangle-cover}).  These
suggestions led to analyses with exponents \(7/4\) and \(12/7\), initially
conditional on further lemmas.  The \(7/4\) analysis matched the window
count in~\cite{CDGKS20} but still used more centers, motivating the
rectangle decomposition.  The \(12/7\) route matched the exponent
in~\cite{CDGKS20}.  The rectangle approach and an adaptation
of Andoni's sparse scheme~\cite{Andoni20}, which we also requested, were
investigated in parallel.  We also asked Sol to replace the general planar
distance oracle of Mozes and Prigan~\cite{MozesPrigan26} with Klein's
multiple-source shortest-path data structure~\cite{Klein05}.
Combining these ideas, with further author
guidance on an early version of \cref{lem:batched-radius-test}, yielded
an \(\Ohtilde(N^{8/5})\)-time algorithm.
The approximation factor was still very large.  Sol then developed the
\((3+\varepsilon)\)-approximation analysis with limited further input
from us, including radius scales separated by factors of
\(1+\varepsilon\).  We also suggested adapting Kuszmaul's
approximation approach~\cite{Kuszmaul19} to remove dependence on the numerical aspect
ratio of the costs.  Sol supplied the detailed constructions and proofs
throughout these developments.

After switching to Astra, we reorganized the still complicated algorithm
into modules corresponding to the technical sections of the present paper.
We requested explicit section contracts and dependencies, guided their
simplification, and asked agents to seek further simplifications autonomously.
We challenged remaining complications and steered the
construction closer to Andoni's scheme~\cite{Andoni20}, including one target-window
family for all distance scales.  These iterations vastly simplified the interfaces
in~\cref{sec:short-pieces-shared-windows,sec:row-comparisons}. While reading the resulting draft, we
proposed many further local improvements, including helper lemmas separating distinct
proof steps.

After learning of the independent parallel work of Mader, Tavasoli, and
Wang~\cite{MaderTavasoliWang26}, we focused more on exposition.
Astra drafted the introduction from our outline and further revised the
technical overview, incorporating detailed feedback from us and reviewer agents.
Under our guidance, Astra also refined the RAM model assumptions and
precision analysis.  The final pass combined extensive manual editing with
AI-assisted refactoring, improving notation, language, and selected statements and proof
arguments.
We take full responsibility for the contents of this~manuscript.

\bibliographystyle{alphaurl}
\bibliography{weighted_cdgks}

\end{document}